\documentclass[a4paper,USenglish,cleveref,autoref,thm-restate]{article}

\usepackage[noend, ruled, resetcount, vlined, linesnumbered]{algorithm2e}
\usepackage{amsmath}
\usepackage{graphicx}
\usepackage{epstopdf}

\usepackage{amsthm}
\usepackage{amssymb} 
\usepackage{hyperref}
\usepackage{cleveref,lineno}
\usepackage{mathtools,amsmath} 
\usepackage{graphicx}
\usepackage{mathrsfs}
\usepackage{thmtools} 
\usepackage{soul,esvect}
\usepackage{epstopdf}
\usepackage{thm-restate,todonotes}
\usepackage{fullpage}
\usepackage{xspace}

\usepackage{fix-cm}

\newtheorem{theorem}{Theorem}
\newtheorem{lemma}{Lemma}

\newtheorem{definition}{Definition}
\newtheorem{corollary}{Corollary}

\newcommand{\dist}[1]{\ensuremath{\text{dist}(#1)}}
\newcommand{\alarm}[1]{\text{alarm}(#1)}
\newcommand{\prev}[1]{\text{prev}(#1)}
\newcommand{\pmid}{\ensuremath{P_\text{mid}}}
\newcommand{\plarge}{\ensuremath{P_\text{large}}}
\newcommand{\psmall}{\ensuremath{P_\text{small}}}
\newcommand{\ppost}{\ensuremath{P_\text{post}}}
\newcommand{\ularge}{\ensuremath{U_\text{large}}}
\newcommand{\usmall}{\ensuremath{U_\text{small}}}

\newcommand{\dsmall}{\ensuremath{D_\text{small}}}
\newcommand{\vora}{\ensuremath{\textsf{Vor}_1}}
\newcommand{\vorb}{\ensuremath{\textsf{Vor}_2}}

\DeclareMathOperator{\update}{\textsc{Update}}

\theoremstyle{definition}

\title{Single-Source Shortest Path Problem in Weighted Disk Graphs}

\author{Shinwoo An\footnote{Pohang University of Science and Technology, Korea. Email: \href{mailto:shinwooan@postech.ac.kr}{\texttt{shinwooan@postech.ac.kr}}} 
 \hspace{50pt}    Eunjin Oh\footnote{Pohang University of Science and Technology, Korea. Email: \href{mailto:eunjin.oh@postech.ac.kr}{\texttt{{eunjin.oh@postech.ac.kr}}}}
 \hspace{50pt}    Jie Xue\footnote{New York University Shanghai, China. Email: \href{mailto:jiexue@nyu.edu}{\texttt{{jiexue@nyu.edu}}}}
 }

\date{}
\begin{document}
\maketitle
\begin{abstract}
    In this paper, we present efficient algorithms for the single-source shortest path problem in weighted disk graphs. A disk graph is the intersection graph of a family of disks in the plane. Here, 
    the weight of an edge is defined as the Euclidean distance
    between the centers of the disks corresponding to the endpoints of the edge. 
    Given a family of $n$ disks in the plane whose radii lie in $[1,\Psi]$ and a source disk, we can compute
    a shortest path tree from a source vertex in the weighted disk graph in $O(n\log^2 n \log \Psi)$ time. 
    Moreover, in the case that the radii of disks are arbitrarily large,
    we can compute a shortest path tree from a source vertex in the weighted disk graph in $O(n\log^4 n)$ time. 
    This improves the best-known algorithm running in $O(n\log^6 n)$ time presented in ESA'23~\cite{kaplan2023unweighted}. 
\end{abstract}

\section{Introduction}
Geometric intersection graphs are a fundamental class of graphs representing spatial relationships among geometric objects. In this paper, we focus on the intersection graph of disks in the plane. 
More formally, for a set $P$ of $n$ points in the plane, where each point $v\in P$ has an associated radius $r_v$,
the \emph{disk graph} $G$ of $P$ is defined as the graph where each vertex corresponds to a point $P$,
and two vertices $u$ and $v$ of $P$ are connected by an edge if and only if the disks with center $u,v$ and radii $r_u,r_v$ intersect. 
For an \emph{edge-weighted} disk graph, the weight of an edge $uv$ is defined as $|uv|$.
In the special case that all radii are the same, the disk graph is also called a \emph{unit disk graph}. 
Disk graphs can be used as a model for broadcast networks: The disks of $P$ represent transmitter-receiver stations with transmission power.  
One can view a broadcast range of a transmitter as a disk.

In this paper, we consider the \emph{single-source shortest path} (SSSP) problem for edge-weighted disk graphs:
Given a set $P$ of points associated with radii and a specified point $s\in P$, compute a shortest path tree of the edge-weighted disk graph of $P$ rooted at $s$. 
One straightforward way to deal with a disk graph is to construct 
the disk graph explicitly, and then run algorithms designed for general graphs. 
However, a disk graph might have complexity $\Theta(n^2)$ in the worst case
even though it can be (implicitly) represented as $n$ disks. 
Therefore, it is natural to seek faster algorithms for a disk graph implicitly represented as its underlying set of disks. 
Besides the SSSP problem, many graph-theoretic problems have much more efficient solutions in disk graphs than in general graphs~\cite{an2023faster,baumann2024dynamic,bonnet2018qptas,cabello2021minimum,espenant2023finding,kaplan2023unweighted,klost2023algorithmic,lokshtanov2023framework}.

\subparagraph{Related work.}
As the single-source shortest path problem is fundamental in computer science, there are numerous work on this problem and its variant for unit disk graphs~\cite{brewer2024improved,cabello2015shortest,chan2016all,chan2019approximate,gao2003well,kaplan_et_al:ACM-SIAM2017,liu2022nearly,wang2020near}.
In the case of unweighted unit disk graphs where all edge weights are one, 
the SSSP problem can be solved in $O(n\log n)$ time, and this is optimal~\cite{cabello2015shortest,chan2016all}.
For edge-weighted unit disk graphs,
the best known exact algorithm for the SSSP problem takes $O(n\frac{\log^2 n}{\log\log n})$ time~\cite{brewer2024improved}, 
and the best known $(1+\varepsilon)$-approximation algorithm takes $O(n \log n+ n \log^2(\frac{1}{\varepsilon}))$ time~\cite{wang2020near}. 
For general disk graphs, the best known exact algorithms for the unweighted and weighted SSSP problem takes $O(n\log^2 n)$ and $O(n\log^6 n)$ time, respectively~\cite{kaplan2023unweighted,klost2023algorithmic}. 

\medskip
Another variant of the problem is the reverse shortest path problem, where the input is an edge-weighted graph $G$, a start vertex, a target vertex, and a threshold value. The problem is to find a minimum value $r$ such that the length of the shortest path from the start vertex to the target vertex in $G_r$ is at most the threshold. Here, $G_r$ is the subgraph of $G$ consisting of edges whose weights are at most $r$. The problem was considered in various metrics in the weighted and unweighted cases. 
The best-known algorithms for both weighted and unweighted unit disk graph metric take $O^*(n^{6/5})$ randomized time~\cite{kaplan2023unweighted,wang2023improved}.
In addition, the best-known algorithm for weighted disk graph metric takes $O^*(n^{5/4})$ randomized time~\cite{kaplan2023unweighted}.

\subparagraph{Our results.}
In this paper, we present two algorithms for the SSSP problem
for edge-weighted disk graphs with $n$ vertices. 
The first algorithm runs in $O(n\log^2 n\log \Psi)$ time, where $\Psi$ is the maximum radius ratio of $P$.
The second algorithm runs in $O(n\log^4 n)$ time. 
This improves the best-known algorithm for this problem running in $O(n\log^6 n)$ time~\cite{kaplan2023unweighted}. 

\subparagraph{Model of computation.} 
In Section~\ref{sec:arbitrary}, our algorithm uses a compressed quadtree. To implement this efficiently, we need to extend the real RAM model~\cite{preparata2012computational} by an \emph{floor function} that rounds a real number down to the next integer~\cite{har2011geometric}. While the floor function is generally considered too powerful, it is widely regarded as reasonable for tasks such as finding the cell of a given level of the grid that contains a given point in constant time~\cite{har2011geometric}. On the other hand, our algorithm in Section~\ref{sec:bounded} operates within the standard real RAM model.

\subparagraph{Preliminaries.}
Throughout this paper, we let $P$ be a set of $n$ points in the plane, where each point $v\in P$ has an associated radius $r_v$, and $G$ be the edge-weighted disk graph of $P$. 
We interchangeably denote $v\in P$ as a point or as a vertex of $G$ if it is clear from the context.
Also, we let $s$ be a source vertex of $P$. We assume that $1\leq r_v$ for all $v\in P$, and thus 
in the case of \emph{bounded radius ratio}, all radii must come from the range $[1,\Psi]$
for a constant $\Psi$. 
The \emph{length} of a path of $G$ is defined as the sum of the weights of the edges contained in the path. The \emph{distance} between $u$ and $v$ is defined as the minimum length among all $u$-$v$ paths of $G$. 
For a vertex $v\in G$, we simply let $d(v)$ denote the distance between $s$ and $v$. 

\section{SSSP on Disk Graphs of Bounded Radius Ratio \texorpdfstring{$\Psi$}{Lg}} \label{sec:bounded}
In this section, we describe our algorithm for the SSSP problem on disk graphs.
Given a set $P$ of $n$ points with associated radii in $[1,\Psi]$, our goal is to compute a shortest path tree from a specified source vertex $s$ in the edge-weighted disk graph $G$ of $P$ in $O(n\log^2 n\log\Psi)$ time. 
More precisely, our goal is to compute $\text{dist}(v)$ and $\prev{v}$ for all vertices $v\in P$ 
such that $\text{dist}(v)=d(v)$, and $\prev{v}$ is the predecessor of $t$ in the shortest $s$-$t$ path. 

\subsection{Sketch of Our Algorithm}
We review the well-known Dijkstra's algorithm that computes a shortest path tree from a source vertex.
Initially, the algorithm sets all dist-values to infinity except a source vertex, and sets $Q=P$. The algorithm sequentially applies the following steps until $Q$ is empty. 

\begin{itemize} \setlength\itemsep{-0.1em}
    \item \textbf{(D1)} Pick the vertex $u$ with smallest dist-value.
    \item \textbf{(D2)} For all neighbors $v\in Q$ of $u$, update $\dist{v}\leftarrow \min\{\dist{v},\dist{u}+|uv|\}$.
    \item \textbf{(D3)} Remove $u$ from $Q$. 
\end{itemize}

In the worst case, Dijkstra's algorithm takes quadratic time since a disk graph can have $\Theta(n^2)$ edges.
In our algorithm, we simultaneously update the dist-values of several vertices by the \textsc{Update} subroutine that was introduced in~\cite{wang2020near}.
For any two vertex set $U$ and $V$ of $G$, $\update(U,V)$ does the following: For all $v\in V$,
\begin{itemize} \setlength\itemsep{-0.1em}
    \item \textbf{(U1)} Compute $u:= \arg\min \{\dist{u}+|uv|\}$ among all $u\in U$ s.t. $uv$ is an edge of $G$.
    \item \textbf{(U2)} Update $\dist{v}\leftarrow\min\{\dist{v}, \dist{u}+|uv|\}$.
\end{itemize}

That is, the subroutine updates the dist-values of all vertices of $V$ using the dist-values of the vertices of $U$. 
The subroutine gets input $U$ and $V$ from a \emph{hierarchical grid}.
For each integer $0\leq i\leq \lceil \log \Psi\rceil$, we use $\Gamma_i$ to denote a \emph{grid} of level $i$, which consists of axis-parallel square cells of diameter $2^i$. 
Then we use $\Gamma$ to denote the union of grid cells of $\Gamma_i$'s.
For a grid cell $c\in \Gamma$, we use $p(c)$ to denote the center of $c$, use $|c|$ to denote the diameter of $c$, and use $\Box_c$ to denote the axis-parallel square of diameter $69|c|$ centered at $p(c)$.
We use $\boxplus_c$ to denote the set of grid cells of $\Gamma_{i-1}, \Gamma_i$ and $\Gamma_{i+1}$ which are contained in $\Box_c$.
We say $c'\in \Gamma_{i-1}$ a \emph{child cell} of $c$ if $c'$ is nested in $c$.
Throughout the paper, we assume that 
no points of $P$ lie on the boundary of any grid cell.
We define two sets of vertices contained in $c$ as follows.
\begin{align*}
    \pmid(c):=\{v\in P\mid v\in c, r_v\in [8|c|, 16|c|)\}, \text{ and }\text{ } 
    \psmall(c):=\{v\in P\mid v\in c, r_v\in [1, 8|c|)\}.
\end{align*}
Intuitively, $\pmid(c)$ consists of vertices contained in $c$ with radii \emph{similar} to $|c|$ while $\psmall(c)$ consists of vertices with radii \emph{smaller} than that of $\pmid(c)$.
Note that $\pmid(c)$ forms a clique in $G$.
For a vertex $v$ in $P$, 
we let $c_v$ be the (unique) cell such that $\pmid(c_v)$ contains $v$. 


\subparagraph{Regular edge and Irregular edge.}
Let $uv$ be an edge of $G$ with $r_u\leq r_v$. We call $uv$ a \emph{regular edge} if $\frac{r_v}{r_u}<2$, and an \emph{irregular edge} if $\frac{r_v}{r_u} \geq 2$.
We can identify all edges of $G$ efficiently by using $\boxplus_c$, $\pmid(c)$ and $\psmall(c)$. 
\begin{lemma} \label{lem:edge-regular}
    Let $uv$ be a regular edge.  
    Then $c_u\in \boxplus_{c_v}$.
    Symmetrically, $c_v\in \boxplus_{c_u}$.
\end{lemma} 
\begin{proof}
    Let $c_v\in \Gamma_i$. Then the radius $r_v$ is in the range $[8|c_v|, 16|c_v|)$. Since $uv$ is a regular edge, $r_u$ is in the range $[4|c_v|, 32|c_v|)$. Consequently, $c_u\in \Gamma_{i-1}\cup \Gamma_i\cup \Gamma_{i+1}$.
    Moreover, as $|uv|$ is at most $r_u+r_v\leq 48|c_v|$ and the diameter of $c_u$ is at most $2|c_v|$, 
    $c_u$ is contained in $\Box_{c_v}$. 
    Thus, we have $c_u\in \boxplus_{c_v}$.
\end{proof}
\begin{lemma} \label{lem:edge-irregular}
    Let $uv$ be an irregular edge with $2r_u\leq r_v$. There is a grid cell $c\in \boxplus_{c_v}$ such that $u\in \psmall(c)$.
\end{lemma} 
\begin{proof}
        Let $c_v\in \Gamma_i$. Since $r_u< r_v$, $|uv|$ is at most $2r_v\leq 32|c_v|$. Hence, $u\in \Box_{c_v}$. Then there is a grid cell $c$ in $\boxplus_{c_v}\cap \Gamma_{i}$ that contains $u$. 
    Then $u\in \psmall(c)$ since $r_u\leq \frac{1}{2}r_v<8|c_v|=8|c|$.
\end{proof}

\subparagraph{Basic strategy of~\cite{wang2020near}.}
We utilize the strategic adaptation of Dijkstra's algorithm in a cell-by-cell manner with $\psmall(\cdot)$ and $\pmid(\cdot)$. 
This strategy was used in~\cite{wang2020near} to design an efficient algorithm for the weighted SSSP problem on unit disk graphs. 
For concreteness, we describe the details of the strategy below.
Let $v$ be the vertex with the smallest dist-value among all vertices not processed so far. As an invariant, we 
shall guarantee that all vertices of $G$ which have been processed
have correct dist-values, and $v$ has a correct dist-value if its predecessor in the shortest $s$-$v$ path has been processed.
We want to process not only $v$ but also all vertices in $\pmid(c_v)$ simultaneously.
However, notice that, at this moment, the vertices in $\pmid(c_v)$ do not necessarily have correct dist-values. 

\medskip
Thus, as a first step, we compute the correct dist-values for all vertices in $\pmid(c_v)$.
If the predecessor $w$ of a vertex $v'\in \pmid(c_v)$
in the shortest $s$-$v'$ path already has the correct dist-value, 
we can simply update the dist-values of $\pmid(c_v)$ by considering all edges whose one endpoint is in $\pmid(c_v)$. 
However, it is possible that $w$ has not been processed yet so far and does not have the correct dist-value.
Fortunately, even in this case, we can show that $w$ has a correct dist-value.
To see this, let $u$ be the predecessor of $w$ in the shortest $s$-$w$ path.
Then $uv'$ cannot be an edge of the graph since otherwise $u$ is the predecessor of $v'$ in the shortest $s$-$v'$ path.
Therefore, $|uv'|>r_{v'}$. As $v,v'\in \pmid(c_v)$, $|uv'|>r_{v'}>|vv'|$.
Hence, $d(u)=d(v')-(|v'w|+|wu|)\leq d(v')-|v'u|<d(v')-|vv'|\leq d(v)$. 
The last inequality follows since the concatenation of $vv'$ and the shortest $s$-$v$ path is longer than the shortest $s$-$v'$ path.
Now $u$ must have been processed since $v$ is the vertex with the smallest dist-value among all vertices not processed so far.
Due to the invariant, $w$ must have the correct dist-value.
Consequently, we can update the dist-values of all vertices of $\pmid(c_v)$ by applying $\update(N(c_v), \pmid(c_v))$ where $N(c_v)$ denotes the set of neighbors of $\pmid(c_v)$ in $G$. 


The second step is to transmit the correct dist-values of $\pmid(c_v)$ into their neighbors in order to satisfy the second invariant.  
Lastly, we remove $\pmid(c_v)$ from the graph.

\subparagraph{Main obstacles and lazy update scheme.}
The complexity of this strategy primarily depends on the cost of searching all neighbors of $\pmid(c_v)$. 
This was not a big deal in~\cite{wang2020near}, as in a unit disk graph, each cell interacts with only a constant number of other cells.
In the case of disk graphs, however, a vertex with radius $\Theta(\Psi)$ can be adjacent to vertices in $\pmid(c)$ for $O(\Psi^2)$ distinct grid cells $c$.
To avoid polynomial dependency on $\Psi$, we handle regular edges and irregular edges separately.
More specifically, we can search all regular edges by considering $\pmid(c)$ for all $c\in \boxplus_{c_v}$ by Lemma~\ref{lem:edge-regular}.
Since $\boxplus_{c_v}$ consists of $O(1)$ cells, we can avoid the polynomial dependency on $\Psi$ through the appropriate use of the \textsc{Update} subroutine.

\medskip
However, the same approach cannot be used to search all irregular edges as up to $\Theta(\Psi^2)$ sets of $\pmid(\cdot)$ may interact with $\pmid(c_v)$ in the worst case.
We address this issue with a novel approach, which we call \emph{lazy update scheme}. 
We say $w$ is a \emph{small neighbor} of a vertex $u$ if $uw$ is an irregular edge and $r_u>r_w$. 
Furthermore, we say $w$ is a \emph{small neighbor} of a cell $c$ if there is a vertex $u\in \pmid(c)$ forming an irregular edge with $w$ and $r_u>r_w$.
In the basic strategy, we transmit the correct dist-value of $v$ to \emph{all} neighbors whenever we process $v$. 
In the lazy update scheme, we postpone the transmission to $\pmid(c)$ for all cells $c$ such that $v$ is a \emph{small} neighbor of $c$. 
We handle several postponed update requests at once at some point.  
More specifically, 
let $V_x$ be the set of small neighbors of $c$ whose dist-values are in the range $[x, x+2|c|]$.
We transmit the dist-values of $V_x$ into $\pmid(c)$ right after 
all vertices of $V_x$ have been processed.
Due to the following lemma, we can bound the number of lazy updates 
into $O(1)$ for each cell $c$, and this enables us to avoid the polynomial dependency on $\Psi$.
\begin{lemma} \label{lem:distance-difference}
    Let $u$ and $u'$ be small neighbors of $c$.
    Then $|d(u)-d(u')|\leq 65|c|$.
\end{lemma}
\begin{proof}
    Suppose $d(u)\leq d(u')$. We show that there is a $s$-$u'$ path whose length is at most $d(u)+65|c|$. Let $v$ and $v'$ be two vertices of $\pmid(c)$ forming an edge with $u$ and $u'$, respectively. 
    Since $v,v'\in c$, we have $|vv'|\leq |c|$. 
    Moreover, both $|uv|$ and $|u'v'|$ is at most $32|c|$ since $r_u<r_v< 16|c|$ and $r_{u'}<r_{v'}< 16|c|$.
    Consider the concatenation of the shortest $s$-$u$ path, $uv$, $vv'$ and $v'u'$. This is a $s$-$u'$ path whose length is at most $d(u)+65|c|$. Since $d(u')$ is the length of the shortest $s$-$u'$ path, $d(u')\leq d(u)+65|c|$. 
\end{proof}

However, this may cause an \emph{inconsistency issue} during the first step.
When we process the vertices of $\pmid(c_v)$, now the
first update does not transmit the dist-values of the small neighbors of $c_v$.
If a predecessor $w$ of $v'\in \pmid(c_v)$ in the shortest $s$-$v'$ path is a small neighbor of $v'$,
it is possible that the lazy update has not occurred even though $w$ has already been processed.
We show that this cannot happen using the following geometric lemma. We postpone the proof of this lemma into Section~\ref{sec:bounded-correct}.
\begin{restatable*}{lemma}{smalltolarge}
\label{lem:smalltolarge}
    Let $u\neq s$ be the predecessor of $v$ in the shortest $s$-$v$ path.
    Then $|uv|\geq |r_v-r_u|$ unless $r_v<r_u$ and $v$ is a leaf in the shortest path tree.    
\end{restatable*}
By this lemma, $d(w)+2|c_{v'}|<d(w)+\frac{1}{4}r_{v'}<d(w)+|r_w-r_{v'}|< d(w)+|wv'|=d(v')$. Subsequently, roughly speaking, the lazy update that transmits dist-values of $w$ into $v'$ occurs in advance when we process the vertices of $\pmid(c_{v'})$.


\subsection{Algorithm} \label{sec:bounded-algo}
We present our algorithm for the SSSP problem on disk graphs with radius ratio $\Psi$.
The goal is to compute $d(v)$ for all $v\in P$.
Initially, we set $\text{dist}(v)$, the \emph{dist-value} of $v$, as infinity for all vertices other than a source vertex $s$, and set $\text{dist}(s)=0$.
Eventually, the algorithm will modify $\text{dist}(v)$ into $d(v)$ for all $v\in P$.
In addition, for each grid cell $c$ of $\Gamma$ with nonempty $\pmid(c)$, we maintain a value $\text{alarm}(c)$ initialized to $\infty$. These alarm values will take care of the moment of the lazy update.
Then we initialize the set $R$ as $P$. 
Here, $R$ is the set of points of $P$ which has not been processed yet. 

\medskip
First, we perform the pre-processing step. 
We compute $d(v)$ for all $v\in P$ adjacent to $s$ in $G$ and set $\text{dist}(v)$ to $|sv|$.
Furthermore, for each grid cell $c$, let $\bar L(c)$ be the set of grid cells $c'$ where $\pmid(c)$ contains a small neighbor of $c'$. 
For technical reasons, we compute a superset $L(c)$ of $\bar L(c)$ which contains all cells $c'$ with $|c|< |c'|$ such that there is an edge between a vertex of $\pmid(c)$ and a vertex of $\pmid(c')$. To see the fact that $\bar L(c)\subset L(c)$, see Lemma~\ref{lem:bounded-lc}.
Furthermore, we set $L(c_s)$ as an empty set where $c_s$ is the grid cell such that the source $s$ is contained in $\pmid(c_s)$.
We will use the information of $L(c)$ to deal with the lazy update. 
Then the algorithm consists of several rounds.
In each round, we check $\text{dist}(v)$ for all $v \in R$ and $\text{alarm}(c)$ for all $c\in \Gamma$.
Subsequently, we find the minimum value $k$ among these values, and 
proceed depending on the type of $k$.
The algorithm terminates when $R$ becomes empty.
We utilize \textsc{Update}$(U,V)$ subroutine for both cases, whose naive implementation is given in Algorithm~\ref{alg-update}.

{\SetAlgoNoLine
\begin{algorithm}[H]\label{alg-update}
    \caption{$\textsc{Update}(U,V)$} 
        \For{$v\in V, u\in U$}{   
           \If{$uv$ \textnormal{is an edge of} $G$}{
                \If{$\textnormal{dist}(v)>\textnormal{dist}(u)+|uv|$}{
                 $\text{dist}(v)\gets \text{dist}(u)+|uv|$;
                 $\text{prev}(v)\gets u$
                 }
            }
        }    
\end{algorithm}}

\medskip
Intuitively, \textsc{Update}$(U,V)$ ensures that if $u\in U$ has correct dist-values and $u$ is the predecessor of $v\in V$ in the shortest $s$-$v$ path, then all such $v$ gets the correct dist-values after the subroutine. 
The detailed implementation of the subroutine is presented in Section~\ref{sec:bounded-complexity}. The rest of this section is devoted to giving detailed instructions on case studies of $k$.

{\SetAlgoNoLine
\begin{algorithm}[H]
    \caption{$\textsc{SSSP-Bounded-Radius-Ratios}(P)$} \label{alg-bounded} 
        $R \leftarrow P$\\
        \While{$R \neq \emptyset$}{
             $k \leftarrow \min(\{\text{dist}(v): v \in R\} \cup \{\text{alarm}(c): c \in \Gamma\})$ \\
            \If{$k=\textnormal{dist}(v)$ \textnormal{for} $v \in R$}{
                 $\textsc{Update}(\bigcup_{c\in \boxplus_{c_v}} \pmid(c), \pmid(c_v))$ \\
                $\textsc{Update}(\pmid(c_v), \bigcup_{c\in \boxplus_{c_v}} \pmid(c)\cup \psmall(c))$ \\
                \For{$c\in L(c_v)$}{
                    \If{$\textnormal{alarm}(c)=\infty$}{
                         $\text{alarm}(c) \leftarrow \text{dist}(v)+2|c|$ 
                    }
                }
                 $R \leftarrow R \setminus \pmid(c_v)$
            }
            \If{$k=\textnormal{alarm}(c)$ \textnormal{for} $c \in \Gamma$}{
                 $\textsc{Update}(\bigcup_{c'\in \boxplus_{c}}\psmall(c'), \pmid(c))$ \\
                  $\text{alarm}(c) \leftarrow \infty$
            }
        }
\end{algorithm}
}

\subparagraph{Case~1: $k=\textnormal{dist}(v)$ for a vertex $v \in R$.}
In this case, we apply $\textsc{Update}(\bigcup_{c\in \boxplus_{c_v}} \pmid(c), \pmid(c_v))$.
Interestingly, we will see later that after this, all vertices in $\pmid(c_v)$ have the correct dist-values. 
Then using the correct dist-values of $\pmid(c_v)$, we update the dist-values of the neighbors of the vertices in $\pmid(c_v)$.
This is done by applying $\textsc{Update}(\pmid(c_v), \bigcup_{c\in \boxplus_{c_v}} \pmid(c)\cup \psmall(c))$ and setting 
$\text{alarm}(c)=\text{dist}(v)+2|c|$ for all $c \in L(c_v)$ with $\text{alarm}(c)=\infty$, where the former takes care of the neighbors of the vertices in $\pmid(c_v)$ connected by regular edges and the small neighbors of $c_v$,
and the latter takes care of the other neighbors.
Finally, we remove $\pmid(c_v)$ from $R$.

\subparagraph{Case~2: $k=\textnormal{alarm}(c)$ for a grid cell $c \in \Gamma$.}
In this case, we shall correct the dist-values of all $v\in \pmid(c)$ such that the predecessor of $v$ is a small neighbor of $v$.
This is done by applying $\textsc{Update}(\bigcup_{c'\in \boxplus_{c}}\psmall(c'), \pmid(c))$.
After this, we reset $\text{alarm}(c)$ to $\infty$.


\subsection{Correctness} \label{sec:bounded-correct}
In this section, we show that the algorithm correctly computes a shortest path tree.
For a vertex $v$, we use $\text{prev}(v)$ to denote the predecessor of $v$ in the shortest path tree.
We shall maintain a simple invariant during the algorithm: $d(v)\leq \text{dist}(v)$ for every $v\in P$. 
\smalltolarge
\begin{proof}
    Let $w$ be the predecessor of $u$ in the shortest $s$-$v$ path. Then $wv$ should not be an edge of $G$ since otherwise $w$ is the predecessor of $v$.
    Then $r_w+r_v<|wv|\leq |wu|+|uv|\leq r_w+r_u+|uv|$. 
    If $r_v \geq r_u$, we are done.
    Otherwise, we assume $r_u>r_v$ and $v$ is not a leaf.
    Let $x$ be a child of $v$ in the shortest path tree.
    As $vx$ is an edge of $G$, $|vx|\leq r_v+r_x$. 
    On the other hand, $r_u+r_x<|ux|$. To see this, $|ux|$ cannot be an edge since otherwise, $u$ is the predecessor of $x$.
    Consequently, 
    $r_u-r_v<(|ux|-r_x)+(r_x-|vx|)=|ux|-|vx|\leq |uv|$. See also Figure~\ref{fig:smalltolarge}.
\end{proof}

\begin{corollary} \label{cor:smalltolarge}
    Suppose the shortest $s$-$v$ path contains an irregular edge $uv$ and $u\neq s$. Then $|uv|\geq \frac{1}{2}r_v$ unless $r_u>r_v$ and $v$ is a leaf on the shortest path tree.
\end{corollary}
\begin{figure}
		\centering
		\includegraphics[width=0.65\textwidth]{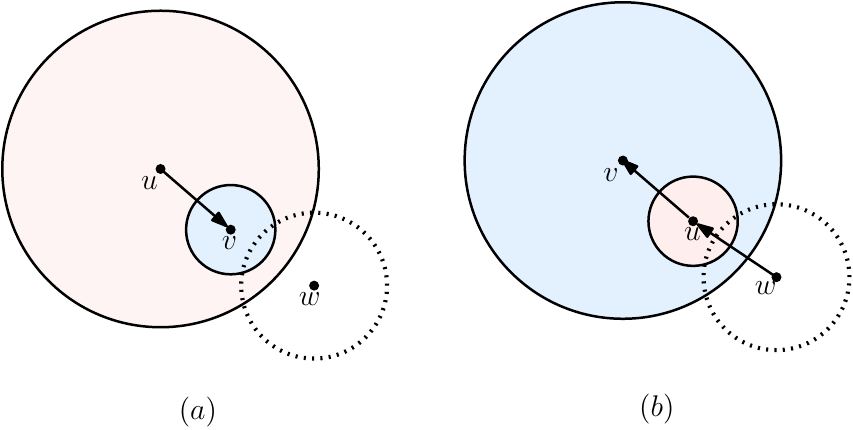}
		\caption{\small (a) If $r_u>r_v$ and $|uv|<r_u-r_v$, all neighbors of $v$ are also neighbors of $u$. (b) If $r_v>r_u$ and $|uv|<r_v-r_u$, the predecessor $w$ of $u$ is neighbor of $v$, and the shortest $s$-$v$ path should contain $wv$.}
		\label{fig:smalltolarge}
\end{figure}

Then we prove the key invariant of our algorithm.
\begin{lemma} The following statements hold during the execution of Algorithm~\ref{alg-bounded}. \label{lem:correct1}
\begin{enumerate} \setlength\itemsep{-0.1em}
    \item[\textnormal{(1)}] At the onset of line 5, $\textnormal{dist}(v)=d(v)$.
    \item[\textnormal{(2)}] After line 5, $\textnormal{dist}(u)=d(u)$ 
    if $u\in P_\textnormal{mid}(c_v)$ is not a leaf on the shortest path tree. 
    \item[\textnormal{(3)}] After line 12, $\textnormal{dist}(u)=d(u)$ if $u\in P_\textnormal{mid}(c)$, 
    $\textnormal{prev}(u)\in \bigcup_{c'\in \boxplus_c} P_\textnormal{small}(c')$, and $\textnormal{prev}(u)\notin R$.
\end{enumerate}
\end{lemma}
\begin{proof}
We apply induction on the index $i$ of the round.
For the base case, the first round of the algorithm is a round of \textbf{Case~1} of $k=\text{dist}(s)$, where $s$ is a source vertex.
The shortest $s$-$u$ path with $u\in \pmid(c_s)$ is $su$ and $\dist{u}$ is already set to $|su|=d(u)$ during the pre-processing.

\medskip
We assume that Lemma~\ref{lem:correct1} holds up to the $(i-1)$-th round.
Suppose the $i$-th round performs line 5.
First we show \textbf{(1)}. 
Let $x$ be the closest vertex to $v$ along the shortest $s$-$v$ path such that $d(x)=\text{dist}(x)$. 
If $x=v$, we are done. Otherwise, we may assume that $x\neq s$ since otherwise, all children of $x$ have the correct dist-values after the pre-processing, which contradicts the choice of $x$.
Then $x\notin R$ since $d(x)<k=\min(\{\text{dist}(v):v\in R\})$, $x\notin R$.
Then by induction hypothesis on the round which removes $x$ from $R$, a child $y$ of $x$ satisfies $d(y)=\text{dist}(y)$ after that round unless $x$ is a small neighbor of $y$. 
If $x$ is a small neighbor of $y$, 
$\text{alarm}(c_y)$ was at most $d(x)+2|c_y|$ after that round.
Since $y\in \pmid(c_y)$, $8|c_y|\leq r_y$. 
Hence, $\text{alarm}(c_y)\leq d(x)+\frac{1}{4}r_y<d(x)+|xy|=d(y)$ due to Corollary~\ref{cor:smalltolarge}.
Thus, $d(y)=\text{dist}(y)$ by induction hypothesis \textbf{(3)} on the round of $k=\text{alarm}(c_y)$.
This contradicts the choice of $x$. Thus, $\text{dist}(v)=d(v)$.

\medskip
Next we show \textbf{(2)}. We may assume $u$ is not a source vertex $s$. Let $w=\text{prev}(u)$.
If $w=s$, then $\dist{u}=d(u)$ due to the pre-processing.
Otherwise, let $x=\prev{w}$. 
We show that $d(w)=\text{dist}(w)$ at the onset of line 5.
If $x=s$, then $d(w)=\text{dist}(w)$ after the pre-processing. Hence, we may assume that $x\neq s$.
Note that the proof of \textbf{(1)} implicitly says that 
for any vertex $v'$ with $d(v') \leq k$, $d(v')=\text{dist}(v')$ at the onset of line 5.
Moreover, if $d(v')<k$, then $v'\notin R$.
Hence, we are enough to consider the case that $d(w)>d(v)$.
We claim that $x\notin R$. Observe that $xu$ is not an edge of $G$ since otherwise, $x$ is the predecessor of $u$. Then $d(x)< d(x)+|xu|-r_u-r_x< d(x)+|xw|+|wu|-r_u=d(u)-r_u<d(v)$, where the last inequality comes from $u,v\in \pmid(c_v)$. 
Hence, $d(x)<k$, $d(x)=\text{dist}(x)$ and $x\notin R$. 
Thus, by induction hypothesis on the round which removes $x$ from $R$,
$d(w)=\text{dist}(w)$ unless $x$ is a small neighbor of $w$ and $k<d(x)+2|c_w|$.
Note that if $x$ is a small neighbor of $w$, $\text{alarm}(c_w)$ is set to at most $d(x)+2|c_w|$ after that round.
We show that $d(x)+2|c_w|<d(v)$. 
Due to Corollary~\ref{cor:smalltolarge}, $d(x)+2|c_w|<(d(w)-|wx|)+\frac{1}{4}r_w\leq d(w)-\frac{1}{4}r_w$.
Due to Lemma~\ref{lem:smalltolarge} in the shortest $s$-$u$ path, $|uw|\geq |r_u-r_w|$. Then $d(w)-\frac{1}{4}r_w = d(u)-|uw|-\frac{1}{4}r_w\leq d(u)-|r_u-r_w|-\frac{1}{4}r_w<d(u)-\frac{1}{4}r_u<d(v)$. Hence, $d(w)=\text{dist}(w)$ due to induction hypothesis \textbf{(3)} on the round of $k=\alarm{c_w}$. 

\medskip
Thus, $d(w)=\text{dist}(w)$ at the onset of line 5.
If $wu$ is a regular edge, \textbf{(2)} holds by line 5. So we may assume that $wu$ is an irregular edge.
By Corollary~\ref{cor:smalltolarge} and the condition that $u$ is not a leaf, $|wu|\geq \frac{1}{2}r_u$.
Then $d(w)=d(u)-|wu|<d(u)-\frac{1}{2}r_u<d(v)$ implies $w\notin R$.
We apply the induction hypothesis on the round which removes $w$ from $R$.
If $u$ is a small neighbor of $w$, the result comes from line 6.
For the other case that $w$ is a small neighbor of $u$,
$\text{alarm}(c_u)$ was at most $d(w)+2|c_u|<d(u)-\frac{1}{4}r_u < d(v)$ after that round. Here, the first inequality comes from Corollary~\ref{cor:smalltolarge}.
Thus, $\text{dist}(u)$ is corrected to $d(u)$ in advance by induction hypothesis on $\textbf{(3)}$. See  Figure~\ref{fig:bounded-correct1} for illustration.

\medskip
Finally we show \textbf{(3)}. Suppose the $i$-th round performs line 12.
By \textbf{(2)} and the fact that $\text{prev}(u)\notin R$, $\text{dist}(\text{prev}(u))=d(\text{prev}(u))$ at the onset of line 12. Then \textbf{(3)} is followed by construction.
\end{proof}

\begin{figure}
		\centering
		\includegraphics[width=0.65\textwidth]{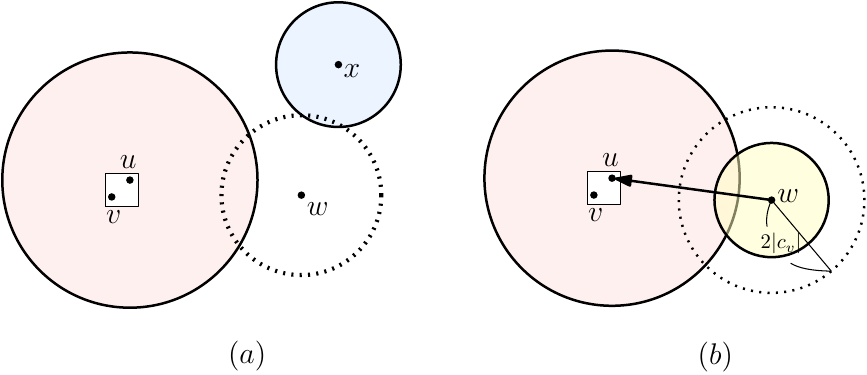}
		\caption{\small Illustrating some points in the proof of Lemma~\ref{lem:correct1}-(2). (a) An edge $xu$ is not an edge of $G$. (b) If $w$ is a small neighbor of $u$. $\text{alarm}(c_v)$ rings before the round of $k=\text{dist}(v)$.}
		\label{fig:bounded-correct1}
\end{figure}

\begin{lemma}\label{lem:bounded-correct}
    Algorithm~\ref{alg-bounded} returns a shortest path tree rooted at $s$.
\end{lemma}
\begin{proof}
Due to line 6 of Algorithm~\ref{alg-bounded} and Lemma~\ref{lem:correct1}-\textbf{(2)}, $d(u)=\text{dist}(u)$ after all rounds of Algorithm~\ref{alg-bounded} unless the predecessor $w$ of $u$ is a small neighbor of $u$.
Even if $w$ is a small neighbor of $u$, Corollary~\ref{cor:smalltolarge} ensures that $\text{alarm}(c_u)$ is at most $d(w)+2|c_u|$ right after the removal of $w$ from $R$.
Since $d(w)+2|c_u|\leq d(w)+\frac{1}{4}r_u<d(u)-\frac{1}{8}r_u\leq \min\{d(v): v\in \pmid(c_u)\}$, the algorithm runs line 11 of $k=\text{alarm}(c_u)$ before $u$ is removed from $R$. Thus, $d(u)=\text{dist}(u)$ by Lemma~\ref{lem:correct1}-\textbf{(3)}.
Finally,
due to the function of \textsc{Update} subroutine, the algorithm eventually computes both the correct dist-value $d(v)$ together with $\text{prev}(v)$ for all $v$. This confirms the correctness.
\end{proof}

\subsection{Time Complexity} \label{sec:bounded-complexity}
In this section, we show that Algorithm~\ref{alg-bounded} can be implemented in $O(n\log^2n\log\Psi)$ time.
For convenience, we assume that the subroutine \textsc{Update}$(U,V)$ takes $O((|U|+|V|)\log^2n)$ time, and analyze the overall running time based on this assumption.
We can efficiently implement the subroutine using standard techniques (ex.~\cite{kaplan2023unweighted, wang2020near}), which we will analyze in the last part of this section.
\begin{lemma} \label{lem:bounded-update}
    The subroutine \textsc{Update}$(U,V)$ takes $O((|U|+|V|)\log^2n)$ time.
\end{lemma}

\subparagraph{Pre-processing.}
In the pre-processing step, we compute $\dist{v}$ for all neighbors $v$ of $s$, and $L(c)$ for all grid cells $c$. 
Here, $L(c)$ denotes the set of cells $c'$ such that $\pmid(c)$ contains a small neighbor of $c'$.
The former part takes $O(n)$ time by checking $|sv|-r_s-r_v$ for all $v\in P$.
We \emph{do not} use the floor function on the real RAM model to implement the hierarchical grid.
\begin{lemma} \label{lem:bounded-grid}
    One can compute a hierarchical grid $\Gamma$, and $P_\textnormal{mid}(c)$ and $P_\textnormal{small}(c)$ for all cells $c\in \Gamma$ in $O(n(\log n+\log \Psi))$ time on the real RAM model. 
\end{lemma}
\begin{proof}
For any $s$-$v$ path $\pi$ in the disk graph, $|sv|\leq 2\sum_{u\in \pi}r_u\leq 2\Psi |\pi|$ where $|\pi|$ denotes the number of vertices contained in $\pi$. Recall that the ultimate goal is to compute the shortest path tree rooted at $s$. Hence, we may assume that  $|sv|\leq 2n\Psi$ for all $v\in P$.
We perform the translation of $P$ to set $s$ as the origin in $O(n)$ time. 
Subsequently,
we compute a square cell $c$ of side length $2n\Psi$ centered at the origin $s$. We set $c$ as the uppermost cell on the hierarchical grid.
Then we utilize the recursive point location query on the hierarchical grid: once we establish that $v$ is contained in $c$, we can compute in $O(1)$ time, a grid cell $c'$ that contains $v$ with $|c'|=|c|/2$ and $c'$ is nested within $c$.
In this way, we can compute $\pmid(c)$ and $\psmall(c)$ for all cells $c$ in $O(n(\log n+\log \Psi))$ time using the standard real RAM model.     
\end{proof}

\begin{lemma} \label{lem:bounded-lc}
    If $\pmid(c)$ contains a small neighbor of $c'$, $c'\in L(c)$. Furthermore, one can compute $L(c)$ for all cells $c$ in $O(n\log n\log\Psi)$ time.
\end{lemma}
\begin{proof}
For the first part, suppose $v\in \pmid(c)$ is a small neighbor of $u\in \pmid(c')$. Since $8|c|\leq r_v<\frac{1}{2}r_u<8|c'|$, $|c|<|c'|$. Hence, $c'\in L(c)$.
For the second part, for a grid cell $c\in \Gamma$, let $N(c)$ denotes the set of grid cells $c'$ such that $2|c|\leq |c'|$ and $c\in \Box_{c'}$. We show $L(c)\subset N(c)$.
Suppose $c'\in L(c)$ and $uv$ be an edge with $u\in \pmid(c')$ and $v\in \pmid(c)$. Then $2|c|\leq |c'|$ is obvious by the definition of $L(c)$.
Moreover, $r_u\geq 8|c'|\geq 16|c|>r_v$ implies $|vu|\leq r_u+r_v<2r_u$.
Then $|vp(c')|\leq |vu|+|up(c')|< 2r_u+|c|< 33|c_u|$ implies $c\in \Box_{c'}$. This verifies $L(c)\subset N(c)$.
Based on this observation, we first compute $N(c)$ in $O(\log \Psi)$ time, by first computing all cells $c'\in \Gamma$ contain $c$, and then for each $c'$, search  $68\times 68$ contiguous cells $c''$ in the same level of the grid centered at $c'$ and check whether $c\subset \Box_{c''}$ or not. 
This computation outputs $N(c)$ because if $c\subset \Box_{c''}$, $68 \times 68$ contiguous cells centered at $c''$ must include $c'$ and $c\subset c'$.

Then for each cell $c'\in N(c)$, we determine whether $c'\in L(c)$ or not as follows.
We compute an additively weighted Voronoi diagram $\textsf{Vor}(c')$ on the points in $\pmid(c')$ with the weight function $w(u)=-r_u$. 
Subsequently, for each $v\in \pmid(c)$, we compute the site in $\pmid(c')$ closest to $v$, denoted by $v(c')$, using $\textsf{Vor}(c')$.
If there is a $(v, v(c'))$ pair such that the distance between $v$ and $v(c')$ in the Voronoi diagram is at most $r_v$, we insert $c'$ into $L(c)$. 

Next, we analyze the time complexity. 
Note that if $c\in N(c_1)\cup N(c_2)$ for two distinct cells $c_1$ and $c_2$, checking $c\in L(c_1)$ and $c\in L(c_2)$ requires a common Voronoi diagram with respect to $\pmid(c)$.
Therefore, we compute an additively weighted Voronoi diagram in $\pmid(c)$ with the weight function $w(u)=-r_u$ for every grid cell $c\in \Gamma$ in advance.
Then the total number of sites stored in all Voronoi diagrams is $O(n)$ since each site is contained in exactly one $\pmid(\cdot)$'s.
In addition, for each $c$, we perform $O(|\pmid(c)|)$ nearest neighbor queries for each $O(\log \Psi)$ Voronoi diagram $\textsf{Vor}(c')$ with $c'\in N(c)$.
Note that an additively weighted Voronoi diagram on $s$ sites can be computed in $O(s\log s)$ time, and the nearest neighbor query on that diagram takes an $O(\log s)$ time~\cite{fortune1986sweepline}.
Therefore, the total time complexity is $\sum_c O(|\pmid(c)|\log|\pmid(c)|)+\sum_c O(|\pmid(c)|\log n\log \Psi)=O(n\log n \log \Psi)$ time.    
\end{proof}

\subparagraph{Time complexity of all rounds of Case~1.}
Let $n_i$ be the number of vertices involved in $\pmid(c_v)\cup \bigcup_{c\in \boxplus_{c_v}}(\pmid(c)\cup \psmall(c))$ in the $i$-th round of line 5-6. If $i$-th round is a round of \textbf{Case~2}, we set $n_i$ to 0. 
The time complexity of this part is then $\Sigma_i O(n_i\log^2 n)$.
We show that each $\pmid(c)$ and $\psmall(c)$ appears in $O(1)$ rounds.
First, since we remove $\pmid(c_v)$ from $R$ after the round $k=\text{dist}(v)$, each $\pmid(c)$ appears exactly once in the term $\pmid(c_v)$.
Suppose $c\in \boxplus_{c'}$. Then $c$ is contained in axis-parallel square of diameter $64|c'|$ centered at $p(c)$ and $|c|\leq 2|c'|$. Conversely, $c'$ is contained in the axis-parallel square of diameter $65|c|$ centered at $p(c)$. Therefore, the number of different $c'$ with $c\in \boxplus_{c'}$ is $O(1)$ for a fixed $c$.
Hence, each $\pmid(c)$ and $\psmall(c)$ appears $O(1)$ rounds through the term $\bigcup_{c\in \boxplus_{c_v}}(\pmid(c)\cup \psmall(c))$.
Note that each vertex $v$ is contained in one $\pmid(c)$ and $O(\log \Psi)$ $\psmall(c)$ for $c\in \Gamma$.
Thus, 
\begin{align}
    \Sigma_i O(n_i\log^2 n) =\Sigma_c O(|\pmid(c)|+\log\Psi |\psmall(c)|) \times O(\log^2 n)= O(n\log^2 n\log \Psi).
\end{align}
\subparagraph{Time complexity of all rounds of Case~2.}
Let $n'_i$ be the number of vertices involved in $\pmid(c)\cup \bigcup_{c'\in \boxplus_{c}}(\psmall(c'))$ in the $i$-th round of line 12. If $i$-th round runs a round of \textbf{Case~1}, we set $n'_i$ to 0. 
The time complexity of this part is then $\Sigma_i O(n'_i\log^2 n)$.
Since $\text{alarm}(c)$ always set to $d(v)+2|c|$ for some $v\in P$ and by Lemma~\ref{lem:distance-difference}, each cell $c$ is referenced by line 11 $O(1)$ times.
Again, there are $O(1)$ grid cells $c'$ such that $\boxplus_{c'}$ contains a fixed cell $c$, and
each $v$ is contained in one $\pmid(c)$ and $O(\log \Psi)$ $\psmall(c)$ sets.
Thus, the total time complexity is
\begin{align}
    \Sigma_i O(n'_i\log^2 n) =\Sigma_c O(|\pmid(c)|+\log\Psi |\psmall(c)|) \times O(\log^2 n)= O(n\log^2 n\log \Psi).
\end{align} 

\subparagraph{Priority queue.}
We maintain two priority queues, one of which stores vertex $v\in R$ with priority $\text{dist}(v)$, and the other queue stores cell $c\in \Gamma$ with priority $\text{alarm}(c)$.
Once the algorithm performs line 3, we peek an element with minimum priority for each queue, and then choose $k$ as the minimum value among them.
The total cost of the queue operations is dominated by the total cost caused by $\update$ subroutines. To see this,
suppose $i$-th round of the algorithm runs a round of $k=\text{dist}(v)$. 
Let $n_i$ be the number of vertices involved in $\pmid(c_v)\cup \bigcup_{c\in \boxplus_{c_v}}(\pmid(c)\cup \psmall(c))$ in line 5-6. 
Throughout lines 3-10, the algorithm changes at most $n_i$ dist-values, changes one alarm-value, and removes $|\pmid(c_v)|\leq n_i$ vertices from the priority queue. 
Thus, the overall time complexity caused by queue operation is $O(n_i\log n)$.
Note that $\update$ subroutine of this round takes $O(n_i\log^2 n)$ time, which dominates the time complexity caused by queue operations.
Similarly, the time complexity of $\update$ subroutine performed in the round of $k=\text{alarm}(c)$ dominates the time complexity caused by queue operations.

\subparagraph{Proof of Lemma~\ref{lem:bounded-update}.}
Recall that \textsc{Update}$(U,V)$ do the following: For all $v\in V$,
\begin{itemize} \setlength\itemsep{-0.1em}
\item \textbf{(U1)} Compute $u:=\arg\min \{\dist{u}+|uv| \}$ among all $u\in U$ s.t. $uv$ is an edge of $G$.
\item \textbf{(U2)} Update $\dist{v}$ to $\min\{\dist{v}, \dist{u}+|uv|\}$.
\end{itemize}

Our implementation is based on several additively weighted Voronoi diagrams with different additively weighted functions. Note that the use of several (additively weighted) Voronoi diagrams to address proximity problems in geometric graphs is not new. See~\cite{kaplan2023unweighted, wang2020near} for examples. 

First, we assign three key values $k_1(u):=\text{dist}(u)+r_u$, $k_2(u):=-r_u$ and $k_3(u):=\text{dist}(u)$ for each vertex $u\in U$.
We store the vertices of $U$ in the balanced binary tree denoted as $T$, according to their ascending order of key values $k_1(\cdot)$.
In addition, for each node $t$ on the binary tree,
we compute two additively weighted Voronoi diagrams $\vora(t)$ and $\vorb(t)$ on the vertices associated to the subtree rooted at $t$, using weight functions by $w_1(v)=k_2(v)$ and $w_2(v)=k_3(v)$, respectively.
This takes $O(|U|\log^2|U|)$ time in total, since we can compute additively weighted Voronoi diagram of $s$ sites in $O(s\log s)$ time~\cite{fortune1986sweepline}.

Next, for each vertex $v$ of $V$, we compute a vertex $u\in U$ of minimum $k_1(u)$ which forms an edge with $v$.
To do this, we traverse $T$ starting from the root. Once we traverse the node $t$, we find the nearest site from $v$ on $\vora(t_\text{left})$ where $t_\text{left}$ is a left child of $t$.
If the distance to that site is greater than $r_v$, which means $r_v+r_u<|uv|$, we traverse to the right child. 
Otherwise, we traverse to the left child. 
Eventually, we reach the leftmost leaf associated with vertex $u$ and $uv$ is an edge of $G$.
Then we compute the minimum value $|vu'|+\text{dist}(u')$ among all vertices $u'$ associated to $T$ subject to the condition $k_1(u)\leq k_1(u')$. 
Since the leaves of $T$ are sorted along $k_1(\cdot)$,
this can be achieved by using $O(\log |U|)$ Voronoi diagrams $\vorb(\cdot)$.
For such $u'$, we update $\dist{v}$ to $\min\{\dist{v}, \dist{u'}+|u'v|\}$, and if $\dist{v}$ is changed, we set $\prev{v}=u'$.
Surprisingly, although $\vorb(\cdot)$ has no information on the graph connectivity, it is sufficient to compute the predecessor of $v$. 
\begin{lemma} \label{lem:update-correct}
    Suppose (1) $|vu_1|+\text{dist}(u_1)>|vu_2|+\text{dist}(u_2)$ and (2) $\text{dist}(u_1)+r_{u_1}\leq \text{dist}(u_2)+r_{u_2}$. If $vu_1$ is an edge of $G$, then $vu_2$ is also an edge of $G$.
\end{lemma}
\begin{proof}
    Since $vu_1$ is an edge of $G$,
    \begin{align}
        |vu_1|\leq r_v+r_{u_1}.
    \end{align}
    Then we obtain the following.
    \begin{align}
        |vu_2|&<|vu_1|+\text{dist}(u_1)-\text{dist}(u_2) \text{ (by (1))} \\
        & <|vu_1|+r_{u_2}-r_{u_1} \text{ (by (2))} \\
        & <r_v+r_{u_2}.
    \end{align}
    Therefore, $vu_2$ is an edge of $G$.
\end{proof}
As we lookup $O(\log|U|)$ Voronoi diagrams for each $v\in V$, updating dist-values for all vertices in $V$ takes $O(|V|\log^2|U|)$ time. Since $|U|\leq n$, the overall implementation takes $O((|U|+|V|)\log^2 n)$ time.

\begin{theorem} \label{thm:boundratio}
    There is an $O(n\log^2n\log\Psi)$-time algorithm that solves the single-source shortest path problem on disk graphs with $n$ vertices and radius ratio $\Psi$.
\end{theorem}

\section{SSSP on Disk Graphs of Arbitrary Radius Ratio} \label{sec:arbitrary}
In this section, we
extend the approach of Section~\ref{sec:bounded} to devise an $O(n\log^4 n)$-time algorithm for the SSSP problem on disk graphs with arbitrary radius ratio. 
Basically, the $O(n\log^2 n\log \Psi)$ term in Theorem~\ref{thm:boundratio} 
came from the total size of $\pmid(\cdot)$ and $\psmall(\cdot)$, which depends on the height (=$\Theta(\log \Psi)$) of the hierarchical grid. 
In the case of an arbitrary radius ratio, we cannot bound the height of the grid. To address this issue, we follow the approach presented in~\cite{baumann2024dynamic}.
More specifically, we use a \emph{compressed quadtree} (See Section~\ref{sec:arbi-prelim}) instead of a hierarchical grid. Although compressed quadtree ensures that the height of the hierarchical structure is independent of the radius ratio, it can be $\Theta(n)$ in the worst case. Hence, we use a \emph{heavy path decomposition} (See Section~\ref{sec:arbi-prelim}) to group quadtree nodes into $O(n\log n)$ disjoint sets $\{\lambda_i\}_{i\in I}$ so that each vertex is contained in $O(\log n)$ different sets. In this way, we can encode the edge information of the disk graph using subquadratic pre-processing time and space.

This approach has been applied to design efficient algorithms on disk graphs of arbitrary radius ratios, such as the dynamic connectivity problem~\cite{baumann2024dynamic} and the unweighted single-source shortest path problem~\cite{klost2023algorithmic}.
By integrating this approach with our lazy update scheme in a more sophisticated way, we can design an $O(n\log^4 n)$-time algorithm for the SSSP problem on weighted disk graphs with arbitrary radius ratios.

\medskip
Throughout this section, we alternatively define the notions of $\pmid(c), \Box_c, \boxplus_c$, and regular edge. Let $h=2^{10}$ and $\alpha > 2\pi/ \arcsin(\frac{1}{100})$ be the integer constants.
\begin{definition} For a grid cell $c$ in $\mathcal Q$, 
    \begin{itemize} \setlength\itemsep{-0.1em}
    \item $\pmid(c):=\{v\in c: r_v\in [h|c|, 2h|c|)\}$.
    \item $\Box_c$:= Axis-parallel square of diameter $8h^2|c|$ centered at $p(c)$.
    \item $\boxplus_c$:= $\{c'\in \mathcal Q: c'\subset \Box_c, |c'|\in [\frac{1}{h}|c|, h|c|]\}$.
    \end{itemize}
\end{definition}
\begin{definition} For an edge $uv$ of $G$ with $r_u\leq r_v$, $uv$ is a \emph{regular edge} if $\frac{r_v}{r_u}<h$, and an irregular edge otherwise.
\end{definition}

Similar to Lemma~\ref{lem:edge-regular}, the modified notions help us to describe regular edges.
\begin{lemma} \label{lem:arbi-regular}
    Let $uv$ be a regular edge. Then $c_u\in \boxplus_{c_v}$.
\end{lemma}
\begin{proof}
       Since $uv$ is a regular edge, $|c_v|<r_u<2h^2|c_v|$.  
    Since $r_u\in [h|c_u|, 2h|c_u|)$ by definition, 
    $\frac{1}{2h}|c_v|<|c_u|<2h|c_v|$. Since $c_u,c_v$ and $h$ are power of two, $|c_u|$ is in range $[\frac{|c_v|}{h}, h|c_v|]$.
    Moreover, 
    $|uv|$ is at most $r_u+r_v <(2h^2+2h)|c_v|<4h^2|c_v|$, and the diameter of $\Box_{c_v}$ is $8h^2|c_v|$.
    Thus, $c_u\subset \Box_{c_v}$ and this implies $c_u\in \boxplus_{c_v}$.
\end{proof}

\subsection{Preliminaries: technical tools} \label{sec:arbi-prelim}
In this section, we first briefly introduce technical items.
Klost~\cite{klost2023algorithmic} implemented an efficient breadth-first search using these components. In contrast, our goal is to implement Dijkstra's algorithm over weighted graphs, which requires more sophisticated update strategies. We conclude this section by outlining the modified lazy update scheme that plays the central role in our algorithm for arbitrary radius ratios in Section~\ref{sec:arbi-algo}.
\subparagraph{Compressed quadtree.}
For an integer $i\geq 0$, let $\Gamma_i$ be a grid of level $i$, which consists of axis-parallel square cells of diameter $2^i$. Among all grid cells of $\bigcup_i \Gamma_i$, let $c$ be the grid cell of the smallest level that contains $P$.
A \emph{compressed quadtree} $\mathcal Q$ is a rooted tree defined as follows.
We start from $c$ as the root of the tree and iteratively expand the tree. More specifically, whenever we process $c\in \Gamma_i$, we compute at most four cells of $\Gamma_{i-1}$ that are nested in $c$, and contain at least one point of $P$. We call these cells \emph{children} of $c$.
If there is exactly one child $c'$, we remove $c$ and connect $c'$ to the parent node of $c$. Then we move on to children (or child) of $c$.
\begin{lemma}[\cite{har2011geometric}]\label{lem:arbi-quadtree}
    In $O(n\log n)$ time, we can compute a compressed quadtree $\mathcal Q$ having $O(n)$ nodes and $O(n)$ height.
\end{lemma}
For our purpose, 
the compressed quadtree is extended to contain all grid cells of $\boxplus_{c_v}$ for all $v\in P$. It also takes $O(n\log n)$ time to compute the extended tree~\cite{baumann2024dynamic}.

\subparagraph{Heavy path decomposition and canonical paths.}
As the height of $\mathcal Q$ is $O(n)$, the direct use of the algorithm from Section~\ref{sec:bounded} is expensive: small neighbors of $v$ are contained in $O(n)$ different $\psmall(\cdot)$ sets, which was $O(\log \Psi)$ in Section~\ref{sec:bounded}. We introduce \emph{heavy-path decomposition} to reduce the number of grid cells that cover all small(large)-neighbors of $v$.
We say an edge $cc'\in \mathcal Q$ is \emph{heavy} if $c'$ is the first child of $c$ in the order of children, where we give an order by the total number of nodes in the subtree rooted at $c'$ among all children of $c$. We say $cc'$ \emph{light} otherwise.
A path of $\mathcal Q$ is \emph{heavy path} if it is the maximal path on $\mathcal Q$ that consists of only heavy edges. Then \emph{heavy path decomposition} is the collection of all heavy paths in $\mathcal Q$. We can efficiently compute heavy path decomposition of small complexity.
\begin{lemma}[\cite{sleator1981data}]
    Let $\mathcal Q$ be the tree of $n$ nodes. Then,
    \begin{enumerate} \setlength\itemsep{-0.1em}
    \item every root-leaf path of $\mathcal Q$ contains $O(\log n)$ light edges,
    \item every node of $\mathcal Q$ lies on exactly one heavy path, and
    \item the heavy path decomposition can be computed in $O(n)$ time.
    \end{enumerate}
\end{lemma}

Then we use the approach of ~\cite{baumann2024dynamic} as a black box. That is, we define a set $\Pi$ of $O(n)$ \emph{canonical paths}, such that every root-node path of $\mathcal Q$ can be uniquely represented by the concatenation of $O(\log n)$ disjoint canonical paths. Also, each canonical path is a subpath of a heavy path in $\mathcal H$, and each node of $\mathcal Q$ is contained in $O(\log n)$ canonical paths. 

\subparagraph{Handling irregular edges.}
Next, we describe a tool for handling irregular edges using canonical paths.
We primarily follow the notion of a proxy graph as presented in~\cite{baumann2024dynamic, klost2023algorithmic}, with slight modifications to the notation for our purpose.
For an irregular edge $uv$ with $r_u<r_v$, we call $u$ a \emph{small neighbor} of $v$, and $v$ a \emph{large-neighbor} of $u$.
For a canonical path $\pi$, we use $c_\pi$ to denote the lowest cell of $\pi$.
Also, we use $|\pi_\ell| (=|c_\pi|)$ and $|\pi_t|$ to denote the diameter of the lowest cell and the topmost cell of $\pi$, respectively.
We then define a set $\mathcal C_\pi$ of $\alpha=O(1)$ congruent cones of radius $3h|\pi_t|$, all sharing the same apex at $p(c_\pi)$. 
We use $\Lambda$ to denote the set of all pairs $(c_\pi, C)$ for all canonical paths $\pi$ and all cones $C\in \mathcal C_\pi$. 
See Figure~\ref{fig:plarge-psmall}(a-b).
 
We say an irregular edge $uv$ is \emph{redundant} if $|uv|<|r_u-r_v|$. Due to Lemma~\ref{lem:smalltolarge}, if a redundant edge $uv$ appears on the shortest path tree, then one endpoint, say $u$, is a leaf on the tree.
Therefore, we can postpone the correction of $\text{dist}(u)$ until later, as no shortest $s$-$w$ path with $w\neq u$ intersects $u$.
Based on this observation, our algorithm handles most redundant edges during post-processing.

\medskip
For a pair $\lambda=(c, C)\in \Lambda$, 
let $r(C)$ be the radius of $C$.
For a vertex $v\in P$, recall that $c_v$ is the grid cell such that $v\in \pmid(c_v)$.
Let $\bar c_v$ be the smallest grid cell on the root-$c_v$ path in $\mathcal Q$ whose diameter is at least $h$ times the diameter of $c_v$. Since the diameters of all grid cells are powers of two, $|\bar c_v|=h|c_v|$.
We let $\Pi_v$ be the set of $O(\log n)$ consecutive canonical paths representing the root-$\bar c_v$ path.
We define three subsets of $P$ with respect to $\lambda$.
\begin{align}
    \psmall(\lambda)&:=\{v\in c\mid \exists \pi\in \Pi_v \text{ with } c=c_\pi\}, \\
    \plarge(\lambda)&:=\{v\in C\mid r_v\in [h|c|, \frac{2}{3}r(C)) \text{ and } |r_v-|vp(c)||<5|c|\}, \text{ and} \\
    \ppost(\lambda)&:=\{v\in C\mid r_v\in [h|c|, \frac{2}{3}r(C)) \text{ and }  5|c|\leq r_v-|vp(c)|\}.
\end{align}
Intuitively, for any irregular edge $uv$ with 
$r_u<r_v$, there is a pair $\lambda$ which encodes $uv$ as $u\in \psmall(\lambda)$ and $v\in \plarge(\lambda)\cup \ppost(\lambda)$.
Moreover, $v\in \plarge(\lambda)$ if $uv$ is a non-redundant edge. See Lemma~\ref{lem:arbi-irregular}.
Later, our algorithm runs Dijkstra's algorithm in a cell-by-cell manner with respect to $\pmid(\cdot), \psmall(\cdot),$ and $\plarge(\cdot)$.
\begin{figure}
		\centering
		\includegraphics[width=0.8\textwidth]{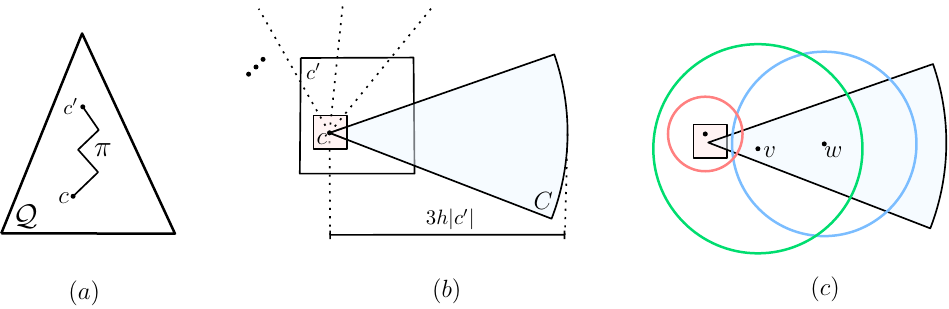}
		\caption{\small (a) Compressed quadtree $\mathcal Q$ and canonical path $\pi$. The lowest cell $c$ and topmost cell $c'$ of $\pi$.  (b) Illustration of $\lambda=(c,C)$. 
        (c) Classification of $\ppost(\lambda)$ and $\plarge(\lambda)$. A disk of $\ppost(\lambda)$(green disk) contains the disks of $\psmall(\lambda)$(red disk), while a disk of $\plarge(\lambda)$(blue disk) may not. 
        }
		\label{fig:plarge-psmall}
\end{figure}
\begin{lemma} \label{lem:arbi-irregular}
    Suppose $v$ is a small neighbor of $u$. 
    There exists a pair $\lambda\in \Lambda$ 
    such that $v\in P_\textnormal{small}(\lambda)$ and $u\in P_\textnormal{large}(\lambda)\cup P_\textnormal{post}(\lambda)$. Moreover, if $uv$ is non-redundant, $u\in P_\textnormal{large}(\lambda)$.
\end{lemma}
\begin{proof}
    We pick a canonical path $\pi$ of $\Pi_v$ such that 
    $h|\pi_\ell|\leq r_u <2h|\pi_t|$. 
    Since $r_u\geq hr_v$ and $r_v\geq h|c_v|=|\bar c_v|$, $r_v\geq h|\bar c_v|$. Hence, 
    such a path $\pi$ always exists.
    Note that $r_v\leq 2h|c_v|=2|\bar c_v|\leq 2|c_\pi|$. 
    Since $v\in c_\pi$, $|uv|\leq r_u+r_v<2|\pi_\ell|+2h|\pi_t|$ and the triangle inequality, 
    $|up(c_\pi)|\leq |uv|+|vp(c_\pi)|\leq 2|\pi_\ell|+2h|\pi_t|+|\pi_\ell|<3h|\pi_t|$. Subsequently, there is a cone $C$ of $\mathcal C_\pi$ that contains $u$.
    Let $\lambda=(c_\pi, C)$. Then $v\in \psmall(\lambda)$ by construction.
    Also, $r_u\in [h|c_\pi|, \frac{2}{3}r(C)]$ since $r_u< 2h|\pi_t|=\frac{2}{3}r(C)$. 
    Since $uv$ is an edge of $G$,
    \begin{align}
        |uv|&\leq r_u+r_v\leq r_u+2|c_\pi|. \text{ Then, } \\
        |up(c_\pi)|&\leq |uv|+|vp(c_\pi)| \text{  }(  \text{by triangle inequality}) \\
        &\leq r_u+3|c_\pi| < r_u+5|c_\pi|. 
    \end{align}
    Hence, $u\in \plarge(\lambda)\cup \ppost(\lambda)$.
    If $uv$ is non-redundant, we have $|uv|\geq r_u-r_v$.
    Then 
    \begin{align}
        |up(c_\pi)|&\geq |uv|-|vp(c_\pi)| \text{  }(  \text{by triangle inequality})\\
        &\geq r_u-r_v-|c_\pi| > r_u-5|c_\pi|. 
    \end{align}
    Then $u\notin \ppost(\lambda)$, which implies $u\in \plarge(\lambda)$. See Figure~\ref{fig:plarge-psmall}(c) for an illustration.
\end{proof}

The total complexities of $\psmall(\lambda)$, $\plarge(\lambda)$, and $\ppost(\lambda)$ are near-linear. 
\begin{lemma} \label{lem:arbi-size}
    $\sum_{\lambda\in \Lambda}|P_\textnormal{small}(\lambda)|=O(n\log n)$ and $\sum_{\lambda\in \Lambda}|P_\textnormal{large}(\lambda)|+|P_\textnormal{post}(\lambda)|=O(n\log n)$.
\end{lemma}
\begin{proof}
    For a fixed $v$, recall that the set $\Pi_v$ consists of $O(\log n)$ consecutive canonical paths divides the root-$\bar c_v$ path. Each canonical path is associated with $O(1)$ congruent cones. In total, there are $O(\log n)$ pairs $\lambda$ such that $\psmall(\lambda)$ contains $v$.

    For the second part, suppose $v\in \plarge(\lambda)\cup \ppost(\lambda)$ with $\lambda=(c, C)$ and $\lambda$ is defined under a canonical path $\pi$.
    We show that $\pi$ contains a grid cell of $\boxplus_{c_v}$.
    By definition of $\plarge(\lambda)$ and $\ppost(\lambda)$, 
    $|vp(c)|< r_v+5|c|\leq 2h|c_v|+5|c_v|<2h^2|c_v|$\footnote{As $v\in \pmid(c)$ implies $r_v\in [h|c|, 2h|c|)$, $|c_v|<|c|$ implies $r_v <  \frac{|c|}{2}\cdot 2h = h|c|$, leads to a contradiction.}.
    Thus, $c$ is contained in $\Box_{c_v}$.
    Subsequently, there is a grid cell $c'\in \boxplus_{c_v}$ with $c\subset c'$ and $|c'|=|c_v|$. 
    Moreover, $|c|\leq |c'|=|c_v|\leq \frac{1}{h}r_v< 2|\pi_t|$ where the last inequality comes from that $r_v< \frac{2}{3}r(C)=2h|\pi_t|$.
    Hence, $c'$ is a cell of $\pi$.
    Thus, the number of different $\lambda$ satisfying $v\in \plarge(\lambda)\cup\ppost(\lambda)$ is at most the number of grid cells in $\boxplus_{c_v}$ multiplied by the maximum number of canonical paths intersect a single cell, which is $O(\log n)$. Hence, the sum of $|\plarge(\lambda)|+|\ppost(\lambda)|$ for all $\lambda\in \Lambda$ is $O(n\log n)$.
\end{proof}

The following statement is a counterpart to Corollary~\ref{cor:smalltolarge}, addressing irregular but non-redundant edges.
\begin{corollary}
\label{cor:arbi-smalltolarge}
    Suppose the shortest $s$-$v$ path contains an irregular edge $uv$ and $v$ is not a leaf. 
    Then $|uv|\geq (1-\frac{1}{h})\max(r_u,r_v)$.
\end{corollary}

\subparagraph{Two-way lazy update scheme.}
Recall that Algorithm~\ref{alg-bounded} updates small neighbors of $v$ once it corrects the dist-values of $\pmid(c_v)$, whereas it postpones the update to (informal) large neighbors of $v$.
In this section, we apply the lazy update scheme for both large and small neighbors of $v$. The main reason is as follows.
Due to Lemma~\ref{lem:correct1}-(2), all vertices of $\pmid(c_v)$ simultaneously get correct dist-values after the call of line 5. This property heavily relies on the geometric property of $\pmid(c_v)$ such that for any $u,w\in \pmid(c_v)$, $|uw|\leq \frac{1}{8}r_u$. However, this property no longer holds for $\plarge(\lambda)$. More specifically, once we handle $v\in \plarge(\lambda)$, there might be a vertex $u\in \plarge(\lambda)$ whose predecessor in the shortest $s$-$u$ has not been processed, and therefore, we cannot get correct dist-values of $\plarge(\lambda)$ at this point.

Hence, one might consider a lazy update scheme to resolve this issue: delay the transmission of dist-values of $\plarge(\lambda)$ in the range $[x, x+f(\lambda)]$ for some function $f$ of $\lambda$ until all dist-values smaller than $x+f(\lambda)$ have been processed. Unfortunately, this naive approach itself is not useful.
Recall that for $v\in \plarge(\lambda)$ with $\lambda=(c,C)$, $r_v$ is in range $[h|c|, \frac{2}{3}r(C)]$ and $r(C)/|c|$ can be very large.
Hence, we cannot partition the dist-values of $\plarge(\lambda)$ into a constant number of ranges $[x, x+f(\lambda)]$, and the number of updates from $\plarge(\lambda)$ to $\psmall(\lambda)$ might be $\Theta(n)$, which leads to quadratic running time in total.

We reduce the running time using the following geometric observation, whose proof is deferred to Section~\ref{sec:arbi-correct}. Let $\lambda=(c,C)$.
\begin{restatable*}{lemma}{arbialarmdown}
\label{lem:arbi-alarm-down}
     Let $v,v'\in \plarge(\lambda)$, $u,u'\in \psmall(\lambda)$, $r_v>r_{v'}$, the shortest $s$-$u$ path contains $vu$, and $v'u'$ is an edge of $G$. 
    Then $d(w) < d(v)+r_v-6|c|$ where $w\in \plarge(\lambda)$ is a predecessor of $u'$ in the shortest $s$-$u'$ path.
\end{restatable*}

Suppose we know the subset $P(\lambda)$ of $\plarge(\lambda)$ whose dist-values have been corrected, and our algorithm will transmit dist-values of $P(\lambda)$ to $\psmall(\lambda)$ right after all dist-values smaller than $d(v)+r_v-6|c|$ have been processed for some $v\in P(\lambda)$.
Then, Lemma~\ref{lem:arbi-alarm-down} guarantees that $P(\lambda)$ contains all vertices of $\plarge(\lambda)$ whose radii are smaller than $r_v$.
Furthermore, after the lazy update, all vertices of $\psmall(\lambda)$ adjacent to $u\in \plarge(\lambda)$ with $r_u\leq r_v$ will have correct dist-values. Hence, the following procedure works.

\begin{itemize} \setlength\itemsep{-0.1em}
    \item \textbf{Step 1.} Maintain a priority queue that stores $v\in P(\lambda)$ with priority $d(v)+r_v-6|c|$.
    \item \textbf{Step 2.} Once we handle the lazy update caused by $d(v)+r_v-6|c|$, we update dist-values of the vertices of $\psmall(\lambda)$ which are adjacent to a vertex $v'\in P(\lambda)$ with $r_{v'}\leq r_v$.
    \item \textbf{Step 3.} After the update, we remove the updated vertices from $\psmall(\lambda)$.
\end{itemize}

Recall that the running time of $\update(U,V)$ is $O((|U|+|V|)\log^2 n)$ time. Although there might $\Theta(n)$ calls of \textsc{Update} for $(\plarge(\lambda),\psmall(\lambda))$, the total size of $|V|$ can be bounded by $\tilde O(n)$ due to \textbf{Step 3} and Lemma~\ref{lem:arbi-size}. 
Notice that $U$ always corresponds to a subset $P(\lambda)$ of $\plarge(\lambda)$, which starts as an empty set and grows as the algorithm progresses. 
Hence, we implement incremental data structures with near-linear construction time by carefully dynamizing \textsc{Update} subroutine, inspired by the spirit of~\cite{bentley1980decomposable}. See Section~\ref{sec:arbi-imp} for the details.
\begin{definition} \label{def:updateinc}
    Let $\textsc{Update-Inc}(W)$ be the incremental data structure with respect to $W$ that supports the following operations:
    \begin{itemize} \setlength\itemsep{-0.1em}
        \item Initialize: set $U=\emptyset$,
       \item Insert($v$): add a vertex $v\in W$ into $U$ together with a dist-value $d(v)$,
        \item Query$(V)$: given a set $V$ of vertices, perform $\textsc{Update}(U,V)$.
    \end{itemize}
\end{definition}
\begin{lemma} \label{lem:arbi-updateinc}
    There is a data structure $\textsc{Update-Inc}(\plarge(\lambda))$ that supports insert operation $O(n\log^4 n)$ time in total, and supports query operation on $V$ in $|V|\cdot O(\log^3 n)$ time. 
\end{lemma}

\subparagraph{Remark.}
We introduced an incremental data structure to handle the $\Theta(n)$ updates from $\plarge(\lambda)$ to $\psmall(\lambda)$. However, due to a similar reason, the number of lazy updates from $\psmall(\lambda)$ to $\plarge(\lambda)$ might be $\Theta(n)$. Fortunately, we can implement this direction of lazy update without using the dynamic data structure using another geometric observation. See Section~\ref{sec:arbi-correct} for details.

\subsection{Algorithm} \label{sec:arbi-algo}
In this section, we present an $O(n\log^4 n)$-time algorithm for the SSSP problem on disk graphs of arbitrary radius ratio. 
The goal is to compute $d(v)$ for all $v\in P$.
Initially, we set $\text{dist}(v)$ as infinity for all vertices other than a source vertex $s$, and set $\text{dist}(s)=0$. 
Ultimately, the algorithm will modify $\text{dist}(v)$ into $d(v)$ for all $v\in P$.
In addition, for each pair $\lambda\in \Lambda$, 
we maintain two values, $\text{alarm-up}(\lambda)$ and $\text{alarm-down}(\lambda)$, together with a priority queue denoted as $Q(\lambda)$. Both $\text{alarm-up}(\lambda)$ and $\text{alarm-down}(\lambda)$ are initialized to $\infty$, and $Q(\lambda)$ is initialized to
an empty queue. 
Then we initialize the set $R$ as the set $P$. 

First we do the pre-processing. 
We compute exact dist-values $d(v)$ for all neighbors of $s$ in $G$ and setting $\text{dist}(v)=d(v)$.
Then for each vertex $v$, we compute the set $L_1(v)$ (and $L_2(v))$ of pairs $\lambda\in \Lambda$ where $v\in\psmall(\lambda)$ (and $v\in \plarge(\lambda)$) and $v$ has a neighbor in $\plarge(\lambda)$ (and $\psmall(\lambda)$). 
Finally, we initialize $\textsc{Updtae-Inc}(\plarge(\lambda))$ for all $\lambda\in \Lambda$.
Then the algorithm consists of several rounds. In each round, we check $\text{dist}(v)$
for all $v\in R$ and 
$\text{alarm-up}(\lambda),\text{alarm-down}(\lambda)$ for all $\lambda\in \Lambda$.
Then we find the minimum value $k$ among them and proceed depending on the type of $k$.
The algorithm moves to the post-processing step when $R$ becomes empty. The algorithm utilizes the \textsc{Update} subroutine from Section~\ref{sec:bounded}.

\subparagraph{Case~1: $k=\textnormal{dist}(v)$ for a vertex $v \in R$.}
In this case, we apply $\textsc{Update}(\bigcup_{c\in \boxplus_{c_v}} \pmid(c), \pmid(c_v))$.
Later we will see that after this, all vertices in $\pmid(c_v)$ have the correct dist-values, except those which are leaves in the shortest path tree.
Then we update the neighbors of $\pmid(c_v)$ using the corrected dist-values. 
This is done by executing subroutine $\textsc{Update}(\pmid(c_v), \bigcup_{c\in \boxplus_{c_v}} \pmid(c))$, setting $\text{alarm-up}(\lambda)=\text{dist}(v)+\frac{1}{4}v(\lambda)$ for every $\lambda\in L_1(v)$ with $\text{alarm-up}(\lambda)=\infty$, and inserting $u$ with the priority $\text{dist}(u)+r_u-6|c|$ into $Q(\lambda)$, inserting $u$ into $\textsc{Update-Inc}(\plarge(\lambda))$ and setting $\text{alarm-down}(\lambda)$ as the minimum priority in $Q(\lambda)$
for every $u\in \pmid(c_v)$ and every $\lambda\in L_2(u)$.
Here, $v(\lambda)$ is the smallest radius of the vertex of $\plarge(\lambda)$ forming an edge with $\pmid(c_v)$. 
Intuitively, the subroutine takes care of the neighbors connected by regular edges, and two types of alarms take care of large neighbors and small neighbors of the vertices of $\pmid(c_v)$, respectively.
Finally, we remove $\pmid(c_v)$ from $R$.

\subparagraph{Case~2: $k=\textnormal{alarm-up}(\lambda)$ for a pair $\lambda\in \Lambda$.} 
In this case, we shall correct the dist-values of all $v\in \plarge(\lambda)$ whose predecessor $u$ is in $\psmall(\lambda)$ with $d(u)<k$.
This is done by applying $\textsc{Update}(\usmall(\lambda, k), \ularge(\lambda, k))$ 
where $\usmall(\lambda, k)$ 
and $\ularge(\lambda, k)$ are computed as follows.
Let $k''<k'$ be the values from the two preceding rounds of $k=\text{alarm-up}(\lambda)$ that are closest to the current round. We set $\usmall(\lambda, k)$ by the vertices of $\psmall(\lambda)$ whose dist-values are in range $[k', k]$. 
Let $w'_k$ and $w''_k$ be the vertices of  $\plarge(\lambda)$ with minimum radii such that they have small neighbors in $\psmall(\lambda)$ with dist-values in range $[k', k]$ and $[-\infty, k'']$, respectively.
We set $\ularge(\lambda, k)$ by the vertices of $\plarge(\lambda)$ whose radius is in range $[r_{w'_k}, r_{w''_k}]$.
Later we show that this is sufficient to correct the dist-values of all $v\in \plarge(\lambda)$ whose predecessor is in $\psmall(\lambda)$.
After this, we reset $\text{alarm-up}(\lambda)$ to $\infty$.

\subparagraph{Case 3: $k=\textnormal{alarm-down}(\lambda)$ for a pair $\lambda\in \Lambda$.}
In this case, we shall correct the dist-values of all $v\in \psmall(\lambda)$ which are adjacent to $u\in \plarge(\lambda)$ with $r_u\leq r_{u_k}$,
where $u_k$ is the vertex stored in $Q(\lambda)$ with minimum priority. 
This is done by the query operation on $\dsmall(\lambda, k)$ to $\textsc{Update-Inc}(\plarge(\lambda, k))$.
Here,
$\dsmall(\lambda, k)$ denotes the set of vertices in $\psmall(\lambda)$ such that (i) form an edge with a vertex of radius at most $r_{u_k}$ contained in $\plarge(\lambda)$, and (ii) have not been considered in a preceding round of $k=\text{alarm-down}(\lambda)$.
Finally, we delete $u_k$ from $Q(\lambda)$ and set $\text{alarm-down}(\lambda)$ to the minimum priority stored in $Q(\lambda)$.

{\SetAlgoNoLine
\begin{algorithm}[H] 
    \caption{$\textsc{SSSP-Arbitrary-Radius-Ratios}(P)$} \label{alg-arbi}
        $R \leftarrow P$ and initialize $\textsc{Update-Inc}(\plarge(\lambda))$ for all $\lambda\in \Lambda$ \\
        \While{$R \neq \emptyset$}{
             $k \leftarrow \min(\{\text{dist}(v): v \in R\} \cup \{\text{alarm-up}(\lambda): \lambda \in \Lambda\}\cup \{\text{alarm-down}(\lambda): \lambda \in \Lambda\})$\\
            \If{$k=\textnormal{dist}(v)$ \textnormal{for} $v \in R$}{
                 $\textsc{Update}(\bigcup_{c\in \boxplus_{c_v}} \pmid(c), \pmid(c_v))$\\
                 $\textsc{Update}(\pmid(c_v), \bigcup_{c\in \boxplus_{c_v}} \pmid(c))$\\
                \For{$\lambda\in L_1(v)$}{
                    \If{$\textnormal{alarm-up}(\lambda)=\infty$}{
                         $\text{alarm-up}(\lambda) \leftarrow \text{dist}(v)+\frac{1}{4}v(\lambda)$ 
                    }
                }

                \For{$u\in P_{\textnormal{mid}}(c_v)$}{
                    \For{$\lambda=(c,C)\in L_2(u)$}{
                         $Q(\lambda)\gets \text{insert}(u, \text{dist}(u)+r_u-6|c|)$\\
                         $\textsc{Update-Inc}(\plarge(\lambda))$.insert$(u)$\\
                         $\text{alarm-down}(\lambda)\gets \text{min-priority}(Q(\lambda))$
                    }
                    
                }
                 $R \leftarrow R \setminus \pmid(c_v)$
            }
            \If{$k=\textnormal{alarm-up}(\lambda)$ \textnormal{for} $\lambda \in \Lambda$}{
                 $\textsc{Update}(\ularge(\lambda, k), \usmall(\lambda, k))$\\
                 $\text{alarm-up}(\lambda) \leftarrow \infty$
            }
            \If{$k=\textnormal{alarm-down}(\lambda)$ \textnormal{for} $\lambda \in \Lambda$}{
                 $\textsc{Update-Inc}(\plarge(\lambda))$.query$(\dsmall(\lambda, k))$\\
                 $\text{alarm-down}(\lambda) \leftarrow \text{min-priority(delete(}Q(\lambda))$
            }
        }
\end{algorithm}}

\subparagraph{Post-processing}
In the post-processing step, we correct the dist-values for the remaining vertices. 
We execute \textsc{Update}$(\ppost(\lambda), \psmall(\lambda))$ for every pair $\lambda$ in $\Lambda$.

\subsection{Correctness} \label{sec:arbi-correct}
In this section, we show that the algorithm from Section~\ref{sec:arbi-algo} correctly computes the shortest path tree. 
Recall that an irregular edge $uv$ is redundant if $|uv|<|r_u-r_v|$. Moreover, if $u=\text{prev}(v)$, then 
$r_u>r_v$ and $v$ is a leaf on the shortest path tree due to Lemma~\ref{lem:smalltolarge}. 
First we show that the vertices in $P_\text{large}(\lambda)$ form a clique. 

\begin{lemma} \label{lem:arbi-clique}
    The vertices of $P_\textnormal{large}(\lambda)$ form a clique in $G$ for every pair $\lambda$.
\end{lemma}
\begin{proof}
     Let $\lambda=(c, C)$ and $o(=p(c))$ be the apex of $C$, and $u,v$ be the vertices of $\plarge(\lambda)$ with $r_v\leq r_u$.
    We enough to show that $|uv|\leq r_u+r_v$.
    Let $x$ be the projection point of $v$ onto the line passing through $o$ and $u$. 
    Since $u,v,x$ are contained in $C$, the angle $\angle vox$ is at most $\frac{2\pi}{\alpha}$.
    Then
    \begin{align}
        |vx|&\leq |vo|\sin(\frac{2\pi}{\alpha})\leq \frac{1}{100}(r_v+5|c|), \\
        |xo|&\leq |vo| \\
        r_u-5|c|&\leq |uo|.
    \end{align}
    If $|uo|\leq |xo|$, we have
    \begin{align}
        |uv|&\leq |ux|+|xv|=(|xo|-|uo|)+|vx|\\
        &\leq  (|vo|-|uo|)+|vx| \\
        &\leq (r_v-r_u+10|c|)+\frac{1}{100}(r_v+5|c|)\\
        &< \frac{1}{100}r_v+11|c| \text{ }(\text{since } r_v\leq r_u)\\
        &<(r_u+r_v) \text{  }(\text{since } 1024|c|\leq r_u)
    \end{align}
    Otherwise, suppose $|xo|\leq |uo|$. Then
    \begin{align}
        |uv|&\leq (|uo|-|xo|)+|vx| \\
        &\leq (|uo|-|vo|)+2|vx| \text{ }(\text{since } |vo|\leq |vx|+|xo|)\\
        &\leq (r_u-r_v+10|c|)+\frac{1}{50}r_v+\frac{1}{10}|c| \\
        &\leq (r_u+r_v) \text{  }(\text{since } 1024|c|\leq r_v).
    \end{align}
    For both cases, we have $|uv|\leq r_u+r_v$.
\end{proof}

\begin{lemma} \label{lem:arbi-correct-main}
The following statements hold during the execution of Algorithm~\ref{alg-arbi}. 
\begin{itemize} \setlength\itemsep{-0.1em}
    \item[\textnormal{(1)}] At the onset of line 5, $\textnormal{dist}(v)=d(v)$ unless $v$ is a leaf.
    \item[\textnormal{(2)}] After line 5, $\textnormal{dist}(u)=d(u)$ 
    if $u\in P_\textnormal{mid}(c_v)$ and $u$ is not a leaf. 
    \item[\textnormal{(3)}] After line 17, $\textnormal{dist}(u)=d(u)$ if $u\in P_\textnormal{large}(\lambda)$, 
    $\textnormal{prev}(u)\in P_\textnormal{small}(\lambda)$, and $d(\textnormal{prev}(u))<k$.
    \item[\textnormal{(4)}] After line 20, 
    let $v$ be the vertex in $Q(\lambda)$ with minimum priority. Then 
    $\textnormal{dist}(u)=d(u)$ if $u\in P_\textnormal{small}(\lambda)$ and $u$ is adjacent to $v'\in P_\textnormal{large}(\lambda)$ with $r_{v'}\leq r_v$. 
\end{itemize}
\end{lemma}
\begin{proof}
The proof of Lemma~\ref{lem:arbi-correct-main} mirrors the logical flow of the proof presented in Lemma~\ref{lem:correct1}.
We apply induction on the index $i$ of the round.
For the base case, the first round of the algorithm is a round of \textbf{Case~1} of $k=\text{dist}(s)$, where $s$ is a source vertex.
    The shortest $s$-$u$ path with $u\in \pmid(c_s)$ is $su$ since $\pmid(c_s)$ forms a clique, and $\dist{u}$ is already set to $|su|=d(u)$ during the pre-processing.

\medskip
Now we assume that Lemma~\ref{lem:arbi-correct-main} holds up to the $(i-1)$-th round.
Suppose $i$-th round performs line 5.
First we show \textbf{(1)}. 
Let $x$ be the closest vertex to $v$ along the shortest $s$-$v$ path such that $d(x)=\text{dist}(x)$. 
If $x=v$, we are done. Otherwise, $x\notin R$ since $d(x)<k=\min(\{\text{dist}(v):v\in R\})$. Then By induction hypothesis on the round which removes $x$ from $R$, a child $y$ of $x$ satisfies $d(y)=\text{dist}(y)$ after line 6 of the round if $xy$ is a regular edge.
If $xy$ is an irregular edge, $xy$ is non-redundant because $v$ is not a leaf.
If $x$ is a small neighbor of $y$, there is a pair $\lambda$ with $x\in\psmall(\lambda)$ and $y\in \plarge(\lambda)$ by Lemma~\ref{lem:arbi-irregular}, and $\text{alarm-up}(\lambda)$ was at most $d(x)+\frac{1}{4}r_y$ after line 9 of that round.
Due to Corollary~\ref{cor:arbi-smalltolarge}, $\text{alarm-up}(\lambda)<d(x)+|xy|=d(y)\leq d(v)$. 
Then by induction hypothesis \textbf{(3)} on the round of $k=\text{alarm-up}(\lambda)$, $d(y)=\text{dist}(y)$.
For the other case that $x$ is a large-neighbor of $y$, there is a pair $\lambda'=(c, C)$ with $x\in\plarge(\lambda')$ and $y\in \psmall(\lambda')$ by Lemma~\ref{lem:arbi-irregular}, and 
$Q(\lambda')$ contains $x$ with priority $d(x)+r_x-6|c|$ after that round. 
Note that $r_y\leq 2|c|$ since $y\in \psmall(\lambda')$. 
Then the priority is at most $d(x)+r_x-r_y < d(x)+|xy|=d(y)\leq d(v)$. 
Here, the first inequality comes from Lemma~\ref{lem:smalltolarge}. Then $d(y)=\dist{y}$ due to induction hypothesis on \textbf{(4)}.
Overall, we have $d(y)=\text{dist}(y)$ and this contradicts the choice of $x$. Thus, $d(v)=\text{dist}(v)$.

\medskip
Then we show \textbf{(2)}. Let $w=\text{prev}(u)$.
We show that $d(w)=\text{dist}(w)$ at the onset of line 5. If $w$ is a source vertex, this is obvious. Therefore, we assume $w\neq s$ and let $x=\text{prev}(w)$.
Note that the proof of \textbf{(1)} implicitly says that for any non-leaf vertex $v'$ with $d(v')\leq k$, $d(v')=\text{dist}(v')$. Moreover, if $d(v')<k$, then $v'\notin R$.
Observe that $xu$ is not an edge of $G$ since otherwise $x$ is the predecessor of $u$. Then $d(x)< d(x)+|xu|-(r_u+r_x)\leq d(x)+|xw|+|wu|-r_u=d(u)-r_u < d(v)$. The last inequality comes from $u,v\in \pmid(c_v)$.
Hence, $d(x)=\text{dist}(x)$ and $x\notin R$.
If $xw$ is a regular edge, $d(w)=\dist{w}$ due to the induction hypothesis \textbf{(2)} on the round which removes $x$ from $R$. 

Otherwise, suppose $xw$ is an irregular edge. Note that $xw$ is non-redundant since $x,w$ are not leaves.
Suppose $x$ is a small neighbor of $w$. Then $x\in \psmall(\lambda)$ and $w\in \plarge(\lambda)$ for some $\lambda\in \Lambda$ due to Lemma~\ref{lem:arbi-irregular}. The alarm $\text{alarm-up}(\lambda)$ was set to at most $d(x)+\frac{1}{4}r_w$ after the round which removes $x$ from $R$, and $d(x)+\frac{1}{4}r_w<d(x)+|xw|-\frac{1}{2}r_w=d(w)-\frac{1}{2}r_w$ since $|xw|\geq (1-\frac{1}{h})r_w$ by Corollary~\ref{cor:arbi-smalltolarge}. Moreover, due to Lemma~\ref{lem:smalltolarge}, $|uw|\geq |r_u-r_w|\geq \frac{1}{2}(r_u-r_w)$. Hence, $d(w)-\frac{1}{2}r_w=d(u)-|uw|-\frac{1}{2}r_w<d(u)-\frac{1}{2}r_u<d(v)$.
Due to induction hypothesis \textbf{(3)} on the round of $k=\text{alarm-up}(\lambda)$, $d(w)=\dist{w}$.
Finally, suppose $x$ is a large-neighbor of $w$. Then $x\in \plarge(\lambda')$ and $w\in \psmall(\lambda')$ for some $\lambda'=(c,C)\in \Lambda$ due to Lemma~\ref{lem:arbi-irregular}. 
Then $Q(\lambda')$ contains $x$ with priority $d(x)+r_x-6|c|$ after that round. Since $r_w$ is at most $2|c|$, the priority is at most $d(x)+r_x-3r_w<d(x)+|xw|-2r_w<d(w)-\frac{1}{2}r_w$. 
Here, the first inequality comes from Lemma~\ref{lem:smalltolarge} to the shortest $s$-$w$ path.
Again by Lemma~\ref{lem:smalltolarge}, $|uw|\geq \frac{1}{2}(r_u-r_w)$.
Then $d(w)-\frac{1}{2}r_w=d(u)-|uw|-\frac{1}{2}r_w<d(u)-\frac{1}{2}r_u<d(v)$.
Hence, by induction hypothesis \textbf{(4)}, $d(w)=\text{dist}(w)$.

Overall, $d(w)=\dist{w}$ at the onset of line 5. We show that $d(u)=\text{dist}(u)$. If $wu$ is a regular edge, $d(u)=\dist{u}$ after the execution of line 5 by construction. If $wu$ is an irregular edge, $|wu|\geq (1-\frac{1}{h})r_u$ by Corollary~\ref{cor:arbi-smalltolarge}, and therefore $d(w)=d(u)-|wu|<d(u)-(1-\frac{1}{h})r_u<d(v)$. Thus, $w\notin R$, and we get $d(u)=\text{dist}(u)$ after line 5 by using similar arguments on induction hypothesis \textbf{(3--4)}. The paragraph below is almost the repetition of the previous paragraph, so the reader may move to the proof of \textbf{(3)} directly.

Suppose $w$ is a small neighbor of $u$. Note that $wu$ is non-redundant as $u$ is not a leaf of the tree. There is a pair $\lambda$ such that $w\in \psmall(\lambda)$ and $u\in \plarge(\lambda)$ by Lemma~\ref{lem:arbi-irregular}. After the round that removes $w$ from $R$, $\text{alarm-up}(\lambda)$ was at most $d(w)+\frac{1}{4}r_u$.
Also, since $|uw|\geq (1-\frac{1}{h})r_u$ by Corollary~\ref{cor:arbi-smalltolarge}, 
$d(w)+\frac{1}{4}r_u<d(w)+|uw|-\frac{1}{2}r_u=d(u)-\frac{1}{2}r_u<d(v)$. Hence, due to induction hypothesis \textbf{(3)} on the round of $k=\text{alarm-up}(\lambda)$, $d(u)=\text{dist}(u)$.
Next, suppose $w$ is a large neighbor of $u$. There is a pair $\lambda'$ such that $w\in \plarge(\lambda')$ and $u\in \psmall(\lambda')$ by Lemma~\ref{lem:arbi-irregular}.
After the round that removes $w$ from $R$, $Q(\lambda')$ contains $w$ with 
priority $d(w)+r_w-6|c|$. Since $r_u$ is at most $2|c|$, the priority is at most $d(w)+r_w-3r_u<d(w)+|wu|-2r_u<d(u)-\frac{1}{2}r_u<d(v)$. Here, the second inequality comes from Lemma~\ref{lem:smalltolarge} to the shortest $s$-$u$ path.
Hence, due to induction hypothesis \textbf{(4)} on the round of $k=\text{alarm-down}(\lambda)$ with $k=d(w)+r_w-6|c|$, $d(u)=\text{dist}(u)$. This completes the proof of \textbf{(2)}.

\medskip
Next, we show \textbf{(3)} in the case that $i$-th round performs line 16. 
Let $w=\text{prev}(u)$.
The difference with Lemma~\ref{lem:correct1}-\textbf{(3)} is that \textsc{Update} takes not $\plarge(\lambda)$ and $\psmall(\lambda)$ but two subsets $\ularge(\lambda, k)\subset \plarge(\lambda)$ and $\usmall(\lambda, k)\subset \psmall(\lambda)$. 
We show that taking these subsets is sufficient for our purpose.
We take $k$ by the minimum value among all rounds of $k=\text{alarm-up}(\lambda)$ such that $d(w)<k$. Then $w\in \usmall(\lambda, k)$ by definition. 
It suffices to show $u\in \ularge(\lambda, k)$.
Let $k''<k'$ be the values from the two preceding rounds of $k=\text{alarm-up}(\lambda)$ that are closest to the current round. 
Note that $d(w)$ is contained in $[k',k]$ by construction. 
Assume to the contrary that $u\notin \ularge(\lambda, k)$.
Then by definition, there is a vertex pair $(x\in \psmall(\lambda), y\in \plarge(\lambda))$ such that $d(x)<k''$, $r_y<r_u$, and $xy$ is an edge of the graph.

We show the contradiction through the geometric analysis.
Let $\lambda=(c,C)$, $o=p(c)$, and $p$ be the projection point of $y$ onto the line passing through $o$ and $u$.
Note that $p$ is contained in $C$ and $|yo|<r_y+5|c|$ since $y\in \plarge(\lambda)$.
Then 
    \begin{align}
        |yp|\leq |yo|\sin(\frac{2\pi}{\alpha})\leq \frac{1}{100}(r_{y}+5|c|),\text{ } |po|<|yo|\leq 5|c|+r_{y}.
    \end{align}
    Note that $r_y-5|c|<r_u-5|c|\leq |uo|$ since $r_y<r_u$ and $u\in \plarge(\lambda)$.
    We show that $|xy|+|yu|\leq |uw|+\frac{1}{8}r_y$ through the case studies on the position of $p$.
    Suppose for the first case that $|uo|\leq |po|$. 
    Then $|pu|=|po|-|uo|\leq (5|c|+r_y)-(r_y-5|c|)=10|c|$. Thus,
    \begin{align}
        |xy|+|yu|&\leq (|xo|+|oy|)+|yu| \\
        &\leq |c|+(|op|+|py|)+(|yp|+|pu|)\leq |c|+|op|+2|py|+|pu| \\
        &=|c|+|ou|+2|py|+2|pu| \text{ }(\text{since } |op|\leq |ou|+|pu|)\\
        &<|ou|+\frac{1}{50}r_{y}+22|c| \\
        &<(|uw|+|wo|)+\frac{1}{50}r_{y}+22|c| \\
        &<|uw|+\frac{1}{8}r_{y}. \text{ }(\text{since } 1024|c|\leq r_y, |wo|\leq |c|)
    \end{align}
    For the other case, suppose $|po|<|uo|$. We simply derive
    \begin{align}
        |xy|+|yu|&\leq (|xo|+|op|+|py|)+(|yp|+|pu|) \\
        &=|c|+|op|+|pu|+2|py| \\
        &=|c|+|ou|+2|py| \text{ }(\text{since } p \text{ lies on } uo)\\
        &< |ou|+\frac{1}{50}r_y+2|c|\\
        &<|uw| +|wo|+\frac{1}{50}r_y+2|c| \\
        &<|uw|+\frac{1}{8}r_{y} \text{ }(\text{since } |wo|\leq |c|).
    \end{align}

    Thus, we have $|xy|+|yu|<|uw|+\frac{1}{8}r_y$ regardless of the position of $p$.
    Now suppose $\frac{1}{4}r_{y}<k'-k''$.
    Then $\frac{1}{4}r_y<k'-k''\leq d(w)-d(x)$ as $d(w')\in [k',k]$ and $d(x)\leq k''$. Hence,
    \begin{align}
        d(x)+|xy|+|yu|<(d(w)-\frac{1}{4}r_y)+|uw|+\frac{1}{8}r_y <d(u)-\frac{1}{8}r_y.
    \end{align}
    Then the concatenation of the shortest $s$-$x$ path, $xy$ and $yu$ is shorter than the shortest $s$-$u$ path, which is a contradiction. 
    For the other case, suppose $k'-k''\leq \frac{1}{4}r_y$.
    Since the round of $k'=\text{alarm-up}(\lambda)$ was executed in advance, there is an edge $x'y'$ with $x'\in \psmall(\lambda), y'\in \plarge(\lambda)$, and $\frac{1}{4}r_{y'}\leq k'-d(x')$.
    We can derive the following formula using a similar argument from the previous case:
    \begin{align}
        |x'y'|+|y'u|&<|uw|+\frac{1}{8}r_{y'}, \text{ then }\\
        d(x')+|x'y'|+|y'u|&<(k'-\frac{1}{4}r_{y'})+|uw|+\frac{1}{8}r_{y'} \\
        &<d(w)+|uw|-\frac{1}{8}r_{y'} \text{ }(\text{since } k'<d(w)) \\
        &=d(u)-\frac{1}{8}r_{y'}.
    \end{align}
    Then the concatenation of the shortest $s$-$x'$ path, $x'y'$ and $y'u$ is shorter than the shortest $s$-$u$ path, which is a contradiction.
    Thus, $u\in \ularge(\lambda, k)$. See also Figure~\ref{fig:arbi-correct1}(a).

\medskip
Finally, we verify \textbf{(4)} in the case that $i$-th round performs line 20. 
First of all, we prove the following geometric property.
\arbialarmdown
\begin{proof}
    Let $\lambda=(c,C)$ and $o$ be the apex of $c$.
First we show that $d(v')<d(v)+r_v-6|c|$.
Since $v,v'\in \plarge(\lambda)$, $r_{v'}-5|c|<|v'o|<r_{v'}+5|c|$ and $|vo|< r_v+5|c|$.
Let $x$ be the projection point of $v'$ onto the line passing through $v$ and $o$. Then $\angle v'ov$ is at most $\frac{2\pi}{\alpha}$. Then
    \begin{align}
        |v'x| &\leq |v'o|\sin(\frac{2\pi}{\alpha})\leq \frac{1}{100}(r_{v'}+5|c|), \text{ and}\\
        |ox|&\geq |v'o|-|v'x| > (r_{v'}-5|c|)-|v'x|>
        \frac{99}{100}r_{v'}-6|c|.
    \end{align}
If $|xo|\leq |vo|$,
 \begin{align}
        |xv|=|ov|-|ox|&\leq (r_v+5|c|)+(6|c|-\frac{99}{100}r_{v'}) \\
        &\leq 11|c|+r_v-\frac{99}{100}r_{v'}.
\end{align}
Since $\plarge(\lambda)$ form a clique in $G$ by Lemma~\ref{lem:arbi-clique}, $d(v')\leq d(v)+|vv'|$.
Hence,
\begin{align}
    d(v')&\leq d(v)+|vx|+|xv'| \\
    &\leq d(v)+(11|c|+r_v-\frac{99}{100}r_{v'})+(\frac{1}{100}(r_{v'}+5|c|)) \\
    &\leq d(v)+r_v +(12|c|-\frac{49}{50}r_{v'}) \\
    &< d(v)+r_v-6|c| \text{ }(\text{since } 1024|c|=h|c|\leq r_{v'}).
\end{align}
For the other case that $|vo|\leq |xo|$, we can derive
\begin{align}
    |xv|=|ox|-|ov| &\leq |v'o|-|ov| \\
    &\leq (r_{v'}+5|c|)-(r_v-5|c|) =r_{v'}-r_v+10|c|.
\end{align}
Similarly,
\begin{align}
    d(v')&\leq d(v)+|xv|+|xv'| \\
    &\leq d(v)+(10|c|+r_{v'}-r_v)+\frac{1}{100}(r_{v'}+5|c|) < d(v)+r_v-6|c| \text{ }(\text{since } r_{v'}\leq r_v). \label{eq:1}
\end{align}
In fact, the above formulae hold for all $w\in \plarge(\lambda)$ with $d_w<\frac{2}{3}r_w$. To see this, in (\ref{eq:1}), $11|c|+(1+\frac{1}{100})\frac{3}{2}r_v-r_v = 11|c|+\frac{49}{100}r_v < r_v-6|c|$ due to $r_v\geq 1024|c|$. Moreover, (\ref{eq:1}) is the only part that uses the fact that $r_{v'}<r_v$.
Hence, $d(w)<d(v)+r_v-6|c|$ when $r_w<\frac{3}{2}r_v$.
Next, we show that $d(w)\leq d(v)+r_v-6|c|$ in the case of $r_w\geq \frac{3}{2}r_v$. 
Since $u'w$ is an edge of the graph,
\begin{align}
    d(u')=d(w)+|wu|\geq d(w)+(r_w-6|c|) \geq d(w)+(\frac{3}{2}r_v-6|c|) \text{ }(\text{since } w\in \plarge(\lambda)).
\end{align}

On the other hand, the length of the concatenation of the shortest $s$-$v'$ path and $v'u'$ is at most $(d(v)+r_v-6|c|)+(r_{v'}+r_{u'})\leq d(v)+2r_v-4|c|$.
Thus, 
\begin{align}
    d(w)+\frac{3}{2}r_v-6|c|<d(u')<d(v)+2r_v-4|c|.
\end{align}
Then $d(w)<d(v)+\frac{1}{2}r_v+|c|<d(v)+r_v-6|c|$ as $1024|c|\leq r_v$. This completes the proof. See also Figure~\ref{fig:arbi-correct1}(c).
\end{proof}
Let $w=\text{prev}(u)$ and $k=d(v)+r_v-6|c|$.
It suffices to consider the case that the $i$-th round is the first round such that $r_{w}\leq r_v$.
Then $u\in \dsmall(\lambda, k)$ is followed by the condition (i,ii) of $\dsmall(\lambda, k)$.
Moreover, due to Lemma~\ref{lem:arbi-alarm-down}, $d(w)<d(v)+r_v-6|c|=k$.
Note that $w\notin R$ due to the choice of $k$.  
Hence, due to induction hypothesis of \textbf{(2)} and $w$ is not a leaf of the shortest path tree, $d(w)=\text{dist}(w)$.
Thus, $w$ has been inserted to $\textsc{Update-Inc}(\plarge(\lambda))$ and therefore $d(u)=\text{dist}(u)$ after the execution of round $k=\text{alarm-down}(\lambda)$. 
\end{proof}

\begin{figure}
		\centering
		\includegraphics[width=0.65\textwidth]{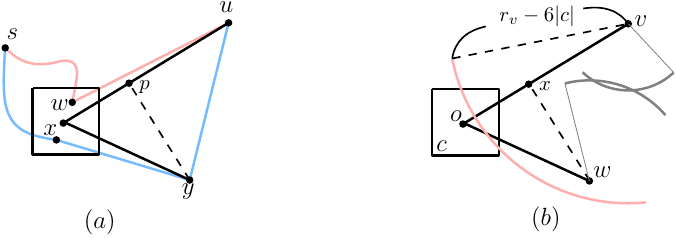}
		\caption{\small Illustration of some points in the proof of Lemma~\ref{lem:arbi-correct-main}.
    (a) In (3), $|yp|$ is small since $\angle you$ is small. Then $|xy|+|yu|$ is not much larger than $|wu|$. The blue path is shorter than the shortest $s$-$u$ path(red path). (b) In (4), $|xw|$ is small. The alarm rings after $w$ has been processed. }
		\label{fig:arbi-correct1}
\end{figure}
The following Lemma verifies the post-processing and therefore, the overall correctness of Algorithm~\ref{alg-arbi}. 
At the onset of the post-processing step, we guarantee the following.
\begin{lemma} \label{lem:arbi-post}
    Suppose the shortest $s$-$v$ path contains $uv$. Unless $v$ is a leaf and 
    $u\in P_\textnormal{post}(\lambda)$ and $v\in P_\textnormal{small}(\lambda)$ for some pair $\lambda$,
    $\textnormal{dist}(v)=d(v)$ after the algorithm completes all rounds.
\end{lemma}
\begin{proof}
      After all rounds of the algorithm, all dist-values of internal vertices are correct. Let $v$ be a leaf on the shortest path tree, and $u$ be its predecessor.
If $uv$ is a regular edge, $d(v)=\text{dist}(v)$ holds by line 6.
If $u$ is a small neighbor of $v$, 
the alarm-up value was at most $d(u)+\frac{1}{4}r_v$ right after the removal of $u$ from $R$. Due to Corollary~\ref{cor:arbi-smalltolarge}, the alarm-up value is at most $d(v)-\frac{1}{h}r_v$, and the algorithm runs line 15 of $k=\text{alarm-up}(\lambda)$ with $u\in \psmall(\lambda)$ and $v\in \plarge(\lambda)$ before $v$ is removed from $R$.
For the other case that $u$ is a large-neighbor of $v$ and there is a pair $\lambda=(c,C)$ with $u\in \plarge(\lambda)$ and $v\in \psmall(\lambda)$, $Q(\lambda)$ contains $u$ with priority $d(u)+r_u-6|c|$ right after the removal of $u$ from $R$. Similarly, we can derive $d(u)+r_u-6|c|<d(v)-\frac{1}{h}r_v$ and $v$ gets correct dist-value before it is removed from $R$. This completes the proof.
\end{proof}

\subsection{Implementation} \label{sec:arbi-imp}
In this section, we show the detailed implementation of Algorithm~\ref{alg-arbi}. While some parts of our algorithm are either previously implemented (ex. \textsc{Update} subroutine) or use well-known data structures (ex. priority queue), we mainly focus on the components consisting of new ideas.

\subparagraph{Pre-processing}
Similar to Section~\ref{sec:bounded-complexity}, we can compute the correct dist-values of all neighbors of a source vertex $s$ in $O(n)$ time.
We show how to compute $\psmall(\lambda)$, $\plarge(\lambda)$ and $\ppost(\lambda)$, together with $L_1(v), L_2(v)$ for all $\lambda\in \Lambda$ and all $v\in P$. 
We compute compressed quadtree, heavy-path decomposition and canonical paths in $O(n\log n)$ time. Then we compute $\psmall(\lambda)$ for all $\lambda$ in $O(n\log n)$ time by computing $\bar c_v$ and $\Pi_v$ in $O(\log n)$ time for each $v\in P$.

We follow the simple procedure presented in~\cite{baumann2024dynamic} to compute $\plarge(\lambda)$ and $\ppost(\lambda)$.
Suppose $v\in \plarge(\lambda)\cup \ppost(\lambda)$ for a pair $\lambda$ originated from a canonical path $\pi$. Then the proof of Lemma~\ref{lem:arbi-size} says that $\pi$ contains a grid cell of $\boxplus_{c_v}$. Also, the number of canonical paths containing a single grid cell is $O(\log n)$.
Based on these observations, for each $v\in P$, we check $O(1)$ cells of $\boxplus_{c_v}$ together with related $O(\log n)$ canonical paths one by one.
While checking a canonical path $\pi$, we check whether $v$ is contained in one of the cones of $\mathcal C_\pi$ in a brute-force manner.
Note that $\mathcal C_\pi$ consists of $\alpha=O(1)$ cones of constant complexity.
Then we check certain conditions of $r_v$. For instance, we verify $v\in \plarge(\lambda)$ or not by checking $r_v\in [h|c|, \frac{2}{3}r(C)]$ and $|r_v-|vp(c)||<5|c|$.
This takes $O(n\log n)$ time in total, since we use the extended compressed quadtree~\cite{baumann2024dynamic} that contains all nodes of $\boxplus_{c_v}$ for all $v\in P$.

To compute $L_1(v)$ and $L_2(v)$, we compute additively weighted Voronoi diagram of weight function $w(v)=-r_v$ for each $\psmall(\lambda)$ and $\plarge(\lambda)$ in $O(n\log^2 n)$ time.
Recall that each $v\in P$ is contained in $O(\log n)$ different $\psmall(\cdot)$ and $\plarge(\cdot)$.
For each $\psmall(\lambda)$ contains $v$, we query on the Voronoi diagram with respect to $\plarge(\lambda)$ to check whether $v$ is a small neighbor of $\plarge(\lambda)$. If this is the case, we add $\lambda$ into $L_1(v)$ (and $L_2(v)$). 
We compute $L_2(v)$ analogously by interchanging the roles of $\psmall(\lambda)$ and $\plarge(\lambda)$. This takes $O(n\log^2 n)$ time.

Then for each $\lambda\in \Lambda$ and for each point $v\in \psmall(\lambda)$, we compute a vertex $v(\lambda)$ of $\plarge(\lambda)$ having the smallest radius which forms an edge with $v$. Later, $v(\lambda)$ will be used to implement \textbf{Case 2}. 
To achieve this, we sort the vertices of $\plarge(\lambda)$ by their ascending order of radii using the balanced binary tree. This takes an $O(n\log^2 n)$ time in total due to Lemma~\ref{lem:arbi-size}.
Moreover, for each node $t$ of the binary tree, we compute an additively weighted Voronoi diagram $\textsf{Vor}(t)$ on the vertices associated with the subtree rooted at $t$, using weight function $w(v)=-r_v$.
Subsequently, for each vertex $v$ of $\psmall(\lambda)$, we traverse the binary tree to find the nearest site on each Voronoi diagram. Note that $v$ forms an edge with the nearest site if and only if the distance to the nearest cite is at most $r_v$. 
This approach allows us to compute $v(\lambda)$ in $O(\log^2 n)$ time for each $v$. Therefore, we can compute $v(\lambda)$ for all $v\in P$ and $\lambda\in \Lambda$ in $O(n\log^3 n)$ time.

\subparagraph{Implementation of Case~2.}
We focus on the key elements appear in \textbf{Case~2}. Let $k''<k'$ represent the values $\text{alarm-up}(\lambda)$ from the two preceding rounds that are closest to the current round $k=\text{alarm-up}(\lambda)$. Then we have to compute:
\begin{itemize} \setlength\itemsep{-0.1em}
    \item $\usmall(\lambda, k)=\{v\in \psmall(\lambda)\mid \text{dist}(v)\in [k', k]\}$
    \item $w'_k, w''_k$: vertices of $\plarge(\lambda)$ with smallest radii which form an edge with $v\in \psmall(\lambda)$ on the condition that $\text{dist}(v)<k''$ ($w'_k$) and $\text{dist}(v)\in [k',k]$ ($w''_k$).
    \item $\ularge(\lambda, k)=\{v\in \plarge(\lambda)\mid r_2\leq r_v<r_1\}$, where $r_1:=r_{w'_k}$, $r_2:=r_{w''_k}$. 
\end{itemize}

For each pair $\lambda$, we manage a dynamic segment tree~\cite{van1993union} of $\psmall(\lambda)$ under the sorted list of $\psmall(\lambda)$ according to dist-values.
In addition, 
each vertex $v$ of $\psmall(\lambda)$ is associated to key value $r_{v(\lambda)}$, and each node $t$ of the tree stores a key value $r(t)$, which is determined by the minimum value $r_{v(\lambda)}$ among all vertices $v$ in the range associated with $t$. Here, $v(\lambda)$ is a vertex of $\plarge(\lambda)$ having minimum radii that are adjacent to $v$.
This additional information does not increases the time complexity of computing the dynamic segment tree.
Whenever a vertex $v\in \psmall(\lambda)$ is removed from $R$, we insert $v$ into the dynamic data structure.
Moreover, we compute a static segment tree of $\plarge(\lambda)$ sorted by their radii.

\medskip
We compute $w'_k, w''_k$ by two range queries on the dynamic segment tree of $\psmall(\lambda)$ with the ranges $[-\infty, k'']$ and $[k', k]$. Note that each range query outputs $O(\log n)$ nodes of the segment tree, and each node stores a key value $r(t)$. We report minimum key value $r(t)$ among those $O(\log n)$ nodes on the segment tree.
Due to the definition of $r(t)$, it is easy to observe that two key values are exactly $r_{w'_k}$ and $r_{w''_k}$.

Then we compute $\usmall(\lambda, k)$ by reporting all vertices of $\psmall(\lambda)$ in range $[k',k]$. Finally, we compute $\ularge(\lambda, k)$ by reporting all vertices of $\plarge(\lambda)$ in range $[r_1, r_2]$ from the segment tree of $\plarge(\lambda)$. 
Consequently, we can compute $\usmall(\lambda, k)$ and $\ularge(\lambda, k)$ in $O(\log n+|\usmall(\lambda, k)|+|\ularge(\lambda, k)|)$ time.
By construction, each vertex of $\psmall(\lambda)$ (and $\plarge(\lambda)$) appears exactly once (and at most twice) on $\usmall(\lambda, k)$ (and $\ularge(\lambda, k)$) throughout all rounds of $k=\text{alarm-up}(\lambda)$.
As the data structure of~\cite{van1993union} can be computed in $O(|\psmall(\lambda)|\log (|\psmall(\lambda)|))$ time, 
the computataion of $\usmall(\lambda, k), \ularge(\lambda, k), w'_k$ and $w''_k$ throughout all rounds of \textbf{Case~2} takes an $O(n\log^2 n)$ time.
Finally, the execution of \textsc{Update} subroutine takes $O(n\log^3 n)$ time in total due to Lemma~\ref{lem:arbi-size}.

\begin{lemma}
    After the pre-processing,
    all rounds of \textbf{Case~2} take $O(n\log^3 n)$ time. 
\end{lemma}

\subparagraph{Implementation of Case 3.}
The following are the key elements that appear in \textbf{Case 3}. 
\begin{itemize} \setlength\itemsep{-0.1em}
    \item $\textsc{Update-Inc}(\plarge(\lambda))$: The incremental data structure specified in Definition~\ref{def:updateinc}.
    \item $u_k$: A vertex in $Q(\lambda)$ with the smallest priority.
    \item $\dsmall(\lambda, k)=\{v\in \psmall(\lambda)\mid r_{v(\lambda)}\leq r_{u_k}, \text{dist}(v) \text{ has not been updated by } \plarge(\lambda) \}$
\end{itemize}

We implement $\textsc{Update-Inc}(\plarge(\lambda))$ as follows.
Whenever $v\in \plarge(\lambda)$ is inserted, we associate three key values $k_1(v):=r_v+\text{dist}(v)$, $k_2(v):=-r_v$ and $k_3(v):=-\text{dist}(v)$.
Furthermore, we maintain $O(\log n)$ binary trees, starting from zero binary trees at the initialization, where the set of leaves of all nodes is identical to the set of inserted vertices, and each binary tree is sorted by $k_1(\cdot)$.
For each node $t$ in a binary tree, we maintain two additively weighted Voronoi diagrams $\textsf{Vor}_1(t), \textsf{Vor}_2(t)$ on the vertices associated with the subtree rooted at $t$,
using weight functions $w_1(v)=k_2(v)$ and $w_2(v)=k_3(v)$, respectively.
When $v$ is inserted, we compute a new binary tree of a single node corresponding to $v$ together with two Voronoi diagrams of a single site. 
If the number of binary trees in the data structure exceeds $\log n$, we find two binary trees $T_1$ and $T_2$ on the structure whose size ratio $\frac{|T_1|}{|T_2|}$ is minimized in the assumption of $|T_2|\leq |T_1|$. We compute the new binary tree together with Voronoi diagrams using the vertices associated with $T_1$ and $T_2$. 
Finally, we remove $T_1$ and $T_2$ so that the number of binary trees in the structure remains $O(\log n)$.

When a query request on $V\subset \psmall(\lambda)$ occurs, 
for each $v\in V$, we traverse $O(\log n)$ binary trees stored in the data structure. We traverse each binary tree $T$ in the same manner with \textsc{Update} subroutine so that we compute a leaf whose corresponding vertex $v$ forms an edge with $u$.
Due to Lemma~\ref{lem:update-correct}, 
$\text{dist}(v)=\text{dist}(u)+|uv|$ if $u=\text{prev}(v)$ and $u\in T$.
Thus, we can compute the predecessor of $u$ by searching $O(\log n)$ binary trees and choose the one vertex $u$ having minimum $\text{dist}(u)$ among them.
\begin{lemma}
     There is a data structure $\textsc{Update-Inc}(\plarge(\lambda))$ that supports insert operation $O(n\log^4 n)$ time in total, and supports query operation on $V$ in $|V|\cdot O(\log^3 n)$ time. 
\end{lemma}
\begin{proof}
    First, we analyze the time complexity of insert operations.
    Since we always insert a single vertex into the data structure, the size of each binary tree of the data structure is always a power of two, and we always merge two binary trees of equal size.
    Recall that from the analysis of \textsc{Update} subroutine, merging two binary trees $T_1$ and $T_2$ takes $O((|T_1|+|T_2|)\log^2n)$ time.
    Throughout the entire insertion, each vertex $v$ contributes $O(\log n)$ merging steps, since the size of the binary tree containing $v$ grows twice for each merging step.
    Therefore, all merging steps take $O(|\plarge(\lambda)|\cdot \log^3 n)=O(n\log^4 n)$ time.
    
    \medskip
    Next, we analyze the time complexity of a query operation. 
    For each $v\in V$, we traverse $O(\log n)$ binary trees. While we traverse a binary tree, we execute $O(\log n)$ nearest neighbor queries through additively weighted Voronoi diagrams stored in a binary tree. 
    Since each nearest neighbor query takes $O(\log n)$ time, the total time complexity is $|V|\cdot O(\log^3 n)$.
\end{proof}

Next, we show how to compute $u_k$ and $\dsmall(\lambda, k)$ efficiently.
We compute $u_k$ in $O(1)$ time using the priority queue $Q(\lambda)$.
In order to compute $\dsmall(\lambda, k)$, we compute a static segment tree of $\psmall(\lambda)$ sorted by $r_{v(\lambda)}$. Again, $v(\lambda)$ is a vertex of $\plarge(\lambda)$ having minimum radii that are adjacent to $v$.
Let $u=u_{k'}$ be the vertex of maximum radius 
such that a round of $k'=\text{alarm-down}(\lambda)$ with $k'=d(u_{k'})+r_{u_{k'}}-6|c|$ executed in advance.
We can maintain $u$ without using extra time.
Then observe that $\dsmall(\lambda, k)$ is exactly the set of vertices $v$ of $\psmall(\lambda)$ such that $r_u < r_{v(\lambda)} \leq r_{u_k}$.
Hence, we can compute $\dsmall(\lambda, k)$ by a single range query on the segment tree with the range $(r_u, r_{u_k}]$.
The time complexity taken by computing $\dsmall(\lambda, k)$ is dominated by the time complexity taken by $\textsc{Update-Inc}$ data structure.

\begin{lemma}
    All rounds of \textbf{Case 3} in the algorithm spend $O(n\log^4 n)$ time in total.
\end{lemma}
\begin{proof}
    Due to the definition of $\dsmall(\lambda, k)$, all vertices of $\psmall(\lambda)$ 
    involved in exactly one $\dsmall(\lambda, k)$. 
    Thus, all rounds of \textbf{Case~3} takes
    \begin{align}
        |\Sigma \psmall(\lambda)| &\times O(\log^3 n)=O(n\log^4 n) \text{ time.}
    \end{align}
\end{proof}
As the incremental data structure $\textsc{Update-Inc}$ is the main bottleneck of our algorithm, we obtain our main result.
\begin{theorem}
    There is an algorithm to solve the single-source shortest path problem on disk graphs in $O(n\log^4 n)$ time.
\end{theorem}

\bibliographystyle{plain}
\bibliography{paper}

\newpage

\end{document}